\documentclass{article}

\usepackage{amsmath,amssymb,graphicx,cite,graphics}

\usepackage{color}	

\newtheorem{definition}{Definition}
\newtheorem{lemma}{Lemma}

\newtheorem{proof}{Proof}

\title{A Mathematical Model for the Dynamics and Synchronization of Cows}

\author{Jie Sun\thanks{Department of Mathematics and Computer Science, Clarkson University,
Potsdam, NY 13699-5815, USA, 
({\tt sunj@clarkson.edu}). 
Current address: Department of Physics \& Astronomy, Northwestern University, Evanston, IL 60208-3112, USA,
({\tt sunj@northwestern.edu}).
}
\and Erik M. Bollt\thanks{Department of Mathematics and Computer Science, Clarkson University,
Potsdam, NY 13699-5815, USA, 
({\tt bolltem@clarkson.edu}).
}
\and Mason A. Porter\thanks{Oxford Centre for Industrial and Applied Mathematics, Mathematical Institute and CABDyN Complexity Centre, University of Oxford, Oxford OX1 3LB, UK, 
({\tt porterm@maths.ox.ac.uk}).
}
\and Marian S. Dawkins\thanks{Department of Zoology, University of Oxford, OX1 3PS, UK, 
({\tt marian.dawkins@zoo.ox.ac.uk}).
}
}

\begin{document}

\maketitle


\begin{abstract}

We formulate a mathematical model for daily activities of a cow (eating, lying down, and standing) in terms of a piecewise affine dynamical system.  We analyze the properties of this bovine dynamical system representing the single animal
and develop an exact integrative form as a discrete-time mapping.  We then couple multiple cow ``oscillators" together to study synchrony and cooperation in cattle herds.  We comment on the relevant biology and discuss extensions of our model.  With this abstract approach, we not only investigate equations with interesting dynamics but also develop interesting
biological predictions.  In particular, our model illustrates that it is possible for cows to synchronize \emph{less} when the coupling is increased.

\end{abstract}



\section*{Keywords}

piecewise smooth dynamical systems, animal behavior, synchronization, cows



\section*{AMS Classification}

37N25, 92D50, 92B25



\pagestyle{myheadings}
\thispagestyle{plain}
\markboth{J. SUN, E.~M. BOLLT, M.~A. PORTER, AND M.~S. DAWKINS}{SYNCHRONIZATION OF COWS}





\section{Introduction} \label{one}

The study of collective behavior---whether of animals, mechanical systems, or simply abstract oscillators---has fascinated a large number of researchers from observational zoologists to pure mathematicians \cite{sync,pikovsky}.  In animals, for example, the study of phenomena such as flocking and herding now involves close collaboration between biologists, mathematicians, physicists, computer scientists, and others \cite{couzinreview,vicsek99,paley07,conradt09}.  This has led to a large number of fundamental insights---for example, bacterial colonies exhibit cooperative growth patterns \cite{ebj94}, schools of fish can make collective decisions \cite{sumpter08}, army ants coordinate in the construction of bridges \cite{couzininpress}, intrinsic stochasticity can facilitate coherence in insect swarms \cite{yates09}, human beings coordinate in consensus decision making \cite{dyer09}, and more.  It has also led to interesting applications, including stabilization strategies for collective motion \cite{scardovi08} and multi-vehicle flocking \cite{chuang07}.

Grazing animals such as antelope, cattle, and sheep derive protection from predators by living in herds \cite{mendl01,estevez07}.  By synchronizing their behavior (i.e., by tending to eat and lie down at the same time), it is easier for the animals to remain together as a herd \cite{rook91,conradt00}.  When out at pasture, cattle are strongly synchronized in their behavior \cite{benham82}, but when housed indoors during the winter, increased competition for limited resources can lead to increased aggression \cite{arave76,nielsen97,mendl01}, interrupted feeding or lying \cite{boe06}, and a breakdown of synchrony \cite{mogensen97}.  There is a growing body of evidence that such disruptions to synchrony (in particular, disruptions to lying down) can have significant effects on cattle production (i.e., growth rate) and cattle welfare \cite{jensen95,hindhede96,mogensen97,fisher02,munk05,gygax07,faerevik08}.  Indeed, synchrony has been proposed as a useful measure of positive welfare in cattle \cite{faerevik08,nap09}, and the European Union regulations stipulate that cattle housed in groups should be given sufficient space so that they can all lie down simultaneously (Council Directive 97/2/EC).  In the winter, cattle have to be housed indoors; space for both lying and feeding is thus limited, and welfare problems can potentially arise because such circumstances interfere with the inherent individual oscillations of cows. 

Although cattle synchronize their behavior if space and resources allow, the mechanism by which they do this is not fully understood \cite{conradt00,nap09}.  In this paper, we examine interacting cattle using a mathematical setting to try to gain an understanding of possible mechanisms.  Viable approaches to studying interacting cows include agent-based models as well as further abstraction via the development and analysis of appropriate dynamical systems to model the cattle behavior.  In a recent dissertation \cite{franz}, B. Franz modified the animal behavior model of Ref.~\cite{couzin} to develop an agent-based model of beef cattle and conduct a preliminary investigation of its synchronization properties.  Given the extreme difficulty of actually understanding the mechanisms that produce the observed dynamics in such models, we have decided instead to take a more abstract approach using dynamical systems.  

Cattle are ruminants, so it is biologically plausible to view them as oscillators.  They ingest plant food, swallow it and then regurgitate it at some later stage, and then chew it again.  During the first stage (standing/feeding), they stand up to graze, but they strongly prefer to lie down and `ruminate' or chew the cud for the second stage (lying/ruminating).  They thus oscillate between two stages.  Both stages are necessary for complete digestion, although the duration of each stage depends on factors such as the nutrient content of the food and the metabolic state of the animal \cite{odrisc09}.\footnote{This oscillating approach to eating is one of the things that made cattle suitable for domestication, as they can eat during the day and then be locked up safely at night to ruminate).}  We thus suppose that each cow is an oscillator, and we choose each oscillator to be a piecewise affine dynamical system in order to incorporate the requisite state-switching behavior in the simplest possible fashion.  Even with this simple model, each individual cow exhibits very interesting dynamics, which is unsurprising given the known complexities of modeling piecewise smooth dynamical systems \cite{Bernardo_2007,scholarpiecewise,Kowalczyk_2005}.  Piecewise smooth systems have been employed successfully in numerous applications---especially in engineering but occasionally also in other subjects, including biology \cite{Glass_1975,Gouze_2002}.  To our knowledge, however, this paper presents the first application of piecewise smooth dynamical systems to animal behavior.  


Our contributions in this paper include the development of a piecewise affine dynamical system model of a cow's eating, lying down, and standing cycles; an in-depth analysis of the mathematical properties of this model; investigation of synchronization in models (which we call \emph{herd models}) produced by coupling multiple copies of the single cow model in a biologically-motivated manner; and a discussion of the biological consequences of our results.  Although our approach is abstract, the present paper is not merely an investigation of equations with interesting dynamics, as we have also developed interesting 
biological predictions.

The rest of this paper is organized as follows.  In Section \ref{two}, we discuss the dynamical system that we use to describe the behavior of a single cow.  We present, in turn, the equations of motion, conditions that describe switching between different states (eating, lying down, and standing), and a discrete representation using a Poincar\'e section.  In Section \ref{three}, we analyze this \emph{single cow model} by studying its equilibrium point, periodic orbits, and bifurcations.  We examine interacting cows in Section \ref{four}.  We present the coupling scheme that we use to construct our \emph{herd equations}, introduce the measure of synchrony that we employ, and examine herd synchrony numerically first for a pair of cows and then for larger networks of cows.  In Section \ref{five}, we comment on our results and briefly discuss variant herd models that can be constructed with different types of coupling.  We then conclude in Section \ref{six} and provide details of our Poincar\'e section and map constructions and analysis in Appendix \ref{app1}.


\section{Single Cow Model}\label{two}

\subsection{Equations of Motion}

We construct a caricature of each cow by separately considering the observable \emph{state} of the cow (eating, lying down, or standing) and its unobservable level of hunger or desire to lie down,
which can each vary between $0$ and $1$.  We also need a mechanism to switch between different states when the level of hunger or desire to lie down exceeds some threshold.  We therefore model each individual cow as a piecewise smooth dynamical system \cite{Bernardo_2007}.

We model the biological status of a single cow by
\begin{equation}
	w=(x,y;\theta)\in [0,1] \times [0,1] \times \Theta\,.
\end{equation}
  The real variables $x$ and $y$ represent, respectively, the extent of desire to eat and lie down of the cow, and
  \begin{equation}
  	\theta\in\Theta=\{\mathcal{E,R,S}\}
  \end{equation}
is a discrete variable that represents the current {state} of the cow (see the equations below for descriptions of the states).  Throughout this paper, we will refer to $\theta$ as a \textit{symbolic variable} or a \textit{state variable}. One can think of the symbolic variable $\theta$ as a switch that triggers different time evolution rules for the other two variables $x$ and $y$.  

We model the dynamics of a single cow in different states using
\begin{eqnarray}\label{eq:singlecow}
	&\mbox{($\cal{E}$)~\textit{Eating state: }}&
	\begin{cases}\label{eq:singlecowE}
		\dot{x} = -\alpha_{2}x\,,\\
		\dot{y} = \beta_{1}y\,.
	\end{cases} \\
	&\mbox{($\cal{R}$)~\textit{Resting state: }}&
	\begin{cases}\label{eq:singlecowR}
		\dot{x} = \alpha_{1}x\,,\\
		\dot{y} = -\beta_{2}y\,.
	\end{cases} \\
	&\mbox{($\cal{S}$)~\textit{Standing state: }}&
	\begin{cases}\label{eq:singlecowS}
		\dot{x} = \alpha_{1}x\,,\\
		\dot{y} = \beta_{1}y\,,
	\end{cases} 
\end{eqnarray}
where the calligraphic letters inside parentheses indicate the corresponding values of $\theta$.  For biological reasons, the parameters $\alpha_1$, $\alpha_2$, $\beta_1$, and $\beta_2$ must all be {positive} real numbers.  They can be interpreted as follows: 
\begin{equation}\nonumber
	\begin{cases}
	\alpha_1: \mbox{rate of increase of hunger}\,,\\
	\alpha_2: \mbox{decay rate of hunger}\,,\\
	\beta_1: \mbox{rate of increase of desire to lie down}\,,\\
	\beta_2: \mbox{decay rate of desire to lie down}\,.
	\end{cases}
\end{equation}

The monotocity in each state (growth versus decay) is the salient feature of the dynamics, and we choose a linear dependence in each case to facilitate analytical treatment.  The piecewise smooth dynamical system describing an individual cow is thus a \emph{piecewise affine dynamical system} \cite{Bernardo_2007}.  As we shall see in the following sections, this simple model is already mathematically interesting.\footnote{Any differential equation whose flow in a given region (increasing versus decreasing) is monotonic in both $x$ and $y$ in all of the states can be treated similarly using the method we describe in Section 2.3 through an appropriate Poincar\'e section.  It is expected to produce qualitatively similar results, as the detailed flow between state transitions is irrelevant once the intersections with Poincar\'e section have been determined.  
}

Additionally, note that we could have added an additional positive parameter $\epsilon\ll{1}$ to each equation to prevent the degeneracy of the $(x,y) = (0,0)$ equilibrium point that occurs for all three equations.\footnote{This degeneracy can also be conveniently avoided by restricting the dynamics of $x$ and $y$ to a region that excludes the point $(0,0)$. We opt for the latter choice (see the next subsection for details).}


\subsection{Switching Conditions}

The dynamics within each state do not fully specify the equations governing a single cow.  To close the bovine equations, we also need switching conditions that determine how the state variable $\theta$ changes.  We illustrate these switching conditions in Fig.~\ref{fig:StateSwitch} and describe them in terms of equations as follows:
\begin{equation}\label{eq:singlecowSwitch}
	\theta\rightarrow
	\begin{cases}
		\cal{E}&\mbox{if $\theta\in\{\cal{R,S}\}$ and $x=1$\,,}\\
		\cal{R}&\mbox{if $\theta\in\{\cal{E,S}\}$ and $x<1$\,, $y=1$\,,}\\
		\cal{S}&\mbox{if $\theta\in\{\cal{E,R}\}$ and $x<1$\,, $y=\delta$ (or $x=\delta\,, y<1$)\,.}
	\end{cases}
\end{equation}
The positive number $\delta<1$ allows the point $(x,y)=(0,0)$ to be excluded from the domain, so that the degenerate equilibrium at that point becomes a so-called \textit{virtual equilibrium point} (i.e., an equilibrium point that is never actually reached by the system) \cite{Bernardo_2007}.

Equations (\ref{eq:singlecowE}, \ref{eq:singlecowR}, \ref{eq:singlecowS}, \ref{eq:singlecowSwitch}) form a complete set of equations describing our \textit{single cow model}. This bovine model is a piecewise smooth dynamical system, to which some important elements of the traditional theory for smooth dynamical systems do not apply, as discussed in depth in the recent book \cite{Bernardo_2007}.

\begin{figure}[htcp]
\centering
\includegraphics[scale=0.5]{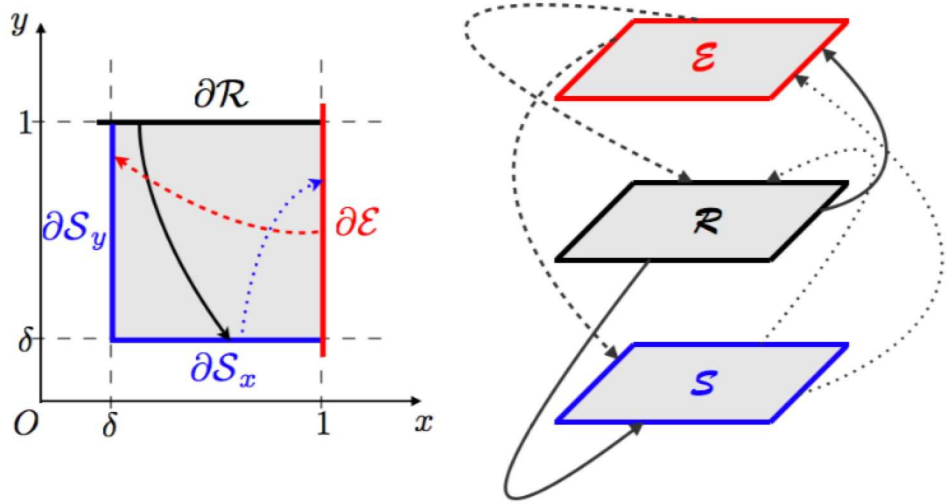}
\caption{(Color online) Switching conditions for the single cow model. In the left panel, we project the set $[\delta,1]\times[\delta,1]\times\Theta$ on $\mathbb{R}^2$, where edges of the square correspond to the borders at which switching occurs. In the right panel, we show the detailed switching situations; an arrow from one edge to another indicates the change of $\theta$ at that edge from one state to the other.  (The arrows with solid curves are the ones that leave state $\mathcal{R}$, those with dashed curves are the ones that leave state $\mathcal{E}$, and those with dotted curves are the ones that leave state $\mathcal{S}$.)
}
\label{fig:StateSwitch}
\end{figure}


\subsection{Discrete Representation}


Although it is straightforward to solve Eqs.~(\ref{eq:singlecowE}, \ref{eq:singlecowR}, \ref{eq:singlecowS}) for the fixed state $\theta$, it is cumbersome to use such a formula to obtain analytical expressions when the flow involves discontinuous changes in $\theta$ (as specified by the switching conditions). Therefore, we instead study the dynamics on the boundaries as discrete maps rather than the flow on the whole domain. We accomplish this by appropriately defining a {Poincar\'e section} \cite{Perko_1996} as the surface
\begin{eqnarray}
	\Sigma&\equiv&\{(x,y;\theta)|x=1\,,\delta\leq{y}\leq{1}\,,\theta=\mathcal{E}\}\cup\{(x,y;\theta)|\delta\leq{x}<{1}\,,y=1\,,\theta=\mathcal{R}\}\nonumber\\
	&=&\mathcal{\partial{E}}\cup\mathcal{\partial{R}}\,,
\end{eqnarray}
which is {transverse} to the flow of Eqs.~(\ref{eq:singlecowE}, \ref{eq:singlecowR}, \ref{eq:singlecowS}) as long as $\alpha_{1,2}$ and $\beta_{1,2}>0$.  (See the Appendix for the proof.)  Furthermore, any flow for which all four of these parameters are positive intersects $\Sigma$ recurrently (again see the Appendix).

Although $\Sigma$ itself is sufficient to construct a Poincar\'e map (we will use $f$ to represent this map on $\Sigma$), it is convenient to consider the discrete dynamics on an {extended Poincar\'e section} $\Sigma'$, which we define by adding the other two boundaries of the projected square to $\Sigma$ to obtain
\begin{eqnarray}
	\Sigma'&\equiv&\Sigma\cup\{(x,y;\theta)|x=\delta\,,\delta\leq{y}<1\}\cup\{(x,y;s)|\delta\leq{x}<1\,,y=\delta\}\nonumber\\
	&=&\mathcal{\partial{E}}\cup\mathcal{\partial{R}}\cup\mathcal{\partial{S}}_y\cup\mathcal{\partial{S}}_x\,,
\end{eqnarray}
where $\partial\mathcal{S}_x$ and 
$\partial\mathcal{S}_y$
are used to represent the sets
$\{(x,y;\theta)|x=\delta,\delta\leq{y}<1\}$ and $\{(x,y;\theta)|\delta\leq{x}<1\,,y=\delta\}$, respectively.  We illustrate the extended Poincar\'e section in the left panel of Fig.~\ref{fig:StateSwitch}.
 
The Poincar\'e map on $\Sigma'$ is given by the discrete dynamics $g:\Sigma'\rightarrow\Sigma'$ derived by solving Eqs.~(\ref{eq:singlecowE}, \ref{eq:singlecowR}, \ref{eq:singlecowS}) with respect to appropriate initial conditions.  As we show in the Appendix, this map is given explicitly by
\begin{eqnarray}\label{casemap}
	g(x=1,\delta\leq{y}\leq{1};\mathcal{E}) = 
	\begin{cases}
	 (y^{\frac{\alpha_2}{\beta_1}},1;\mathcal{R})\,, & 
		\mbox{~if~}y\geq\delta^{\frac{\beta_1}{\alpha_2}}\,,\mbox{~~case $(a)$}\,;\\
	(\delta,\delta^{-\frac{\beta_1}{\alpha_2}}y;\mathcal{S})\,, &
		\mbox{~if~}y<\delta^{\frac{\beta_1}{\alpha_2}}\,,\mbox{~~case $(b)$}\,;\\
	\end{cases}\nonumber\\
	g(\delta\leq{x}<1,y=1;\mathcal{R}) = 
	\begin{cases}
	(1,{x}^{\frac{\beta_2}{\alpha_1}};\mathcal{E})\,, &
		\mbox{~if~}x\geq\delta^{\frac{\alpha_1}{\beta_2}}\,,\mbox{~~case $(c)$}\,;\\
	(\delta^{-\frac{\alpha_1}{\beta_2}}x,\delta;\mathcal{S})\,, &
		\mbox{~if~}x<\delta^{\frac{\alpha_1}{\beta_2}}\,,\mbox{~~case $(d)$}\,;\\	
	\end{cases}\nonumber\\
	g(x=\delta,\delta\leq{y}<1;\mathcal{S}) = 
	\begin{cases}
		(1,\delta^{-\frac{\beta_1}{\alpha_1}}y;\mathcal{E})\,, &
			\mbox{~if~}y\leq{\delta}^{\frac{\beta_1}{\alpha_1}}\,,\mbox{~~case $(e)$}\,;\\
		(y^{-\frac{\alpha_1}{\beta_1}}\delta,1;\mathcal{R})\,, &
			\mbox{~if~}y>{\delta}^{\frac{\beta_1}{\alpha_1}}\,,\mbox{~~case $(f)$}\,;\\
	\end{cases}\nonumber\\
	g(\delta<x<1,y=\delta;\mathcal{S}) = 
	\begin{cases}
		(1,x^{-\frac{\beta_1}{\alpha_1}}\delta;\mathcal{E})\,, &
			\mbox{~if~}x\geq{\delta}^{\frac{\alpha_1}{\beta_1}}\,,\mbox{~~case $(g)$}\,;\\
		(\delta^{-\frac{\alpha_1}{\beta_1}}x,1;\mathcal{R})\,, &
			\mbox{~if~}x<{\delta}^{\frac{\alpha_1}{\beta_1}}\,,\mbox{~~case $(h)$}\,.\\
	\end{cases}
\end{eqnarray}
In Fig.~\ref{fig:FlowDemo}, we show all possible mappings on $\Sigma'$ and, in particular, illuminate all of the possible cases in (\ref{casemap}). The Poincar\'e map $f:\Sigma\rightarrow \Sigma$ can be obtained from $g$ (see the discussion in the Appendix).  

\begin{figure}[htcp]
\centering
\includegraphics[scale=0.5]{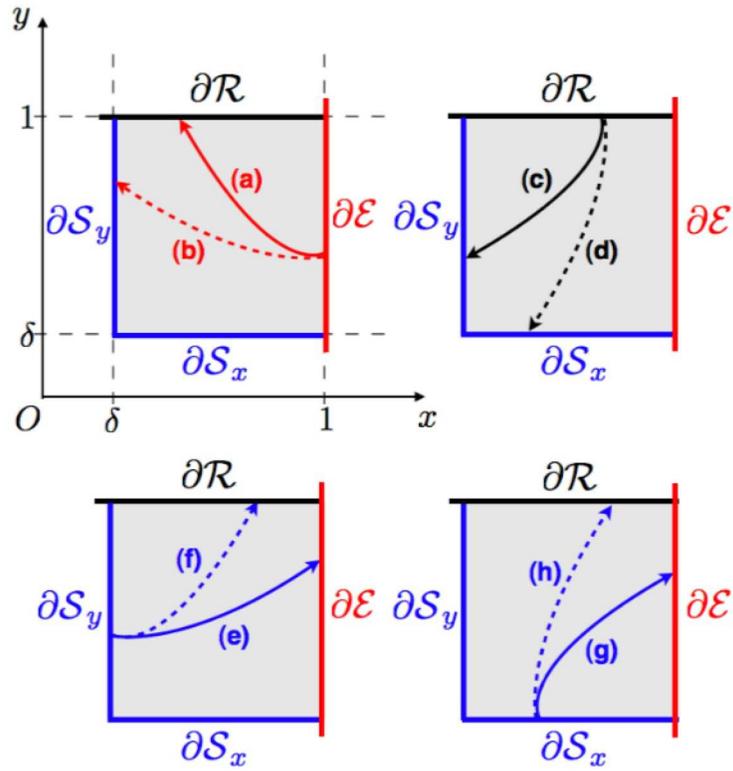}
\caption{(Color online) All of the possible rules for determining the discrete dynamics on $\Sigma'$ that are derived from the original system. For example, from $\theta=\mathcal{E}$, the flow is going to either hit the horizontal $y=1$, which triggers the state $\theta\rightarrow\mathcal{R}$ [case $(a)$], or hit the vertical $x=\delta$, resulting in the transition $\theta\rightarrow\mathcal{S}$ [case $(b)$]. The other three panels similarly demonstrate the other switching possibilities for the variable $\theta$.
}
\label{fig:FlowDemo}
\end{figure}



\section{Analysis of the Single Cow Model} \label{three}

In this section, we summarize a few properties of the single cow model in terms of the discrete dynamics $f$ on $\Sigma$.  Specifically, we give analytical results for the emergence and stability of the fixed point (which is unique) and the period-two orbits on $\Sigma$. We include detailed derivations of these results in the Appendix. We summarize these results in Table~\ref{TableSummary}.

For convenience, we assume that the cow is initially in the state $\mathcal{E}$ with $x = 1$ and $\delta \leq y \leq 1$. If the cow were to start in other situations, it would eventually come to this state. Furthermore, we have chosen to assign the state value $\theta=\mathcal{S}$ to the point $(x,y)=(1,1)$ for as a tie-breaker. Similarly, $\theta=\mathcal{E}$ at $(x,y)=(1,\delta)$ and $\theta=\mathcal{R}$ at $(x,y)=(\delta,1)$, in accordance with Eq.~(\ref{eq:singlecowSwitch}).


\subsection{Fixed Point}

The only possible fixed point on $\Sigma$ is the corner point $(x,y;s)=(1,1;\mathcal{E})$. This fixed point is asymptotically stable if and only if the parameters satisfy
\begin{equation}
	\frac{\alpha_2}{\alpha_1}\cdot\frac{\beta_2}{\beta_1}<1\,.
\end{equation}
Additionally, (when the above condition holds) numerical simulations indicate that the basin of attraction of this fixed point seems to be the entire domain.



\subsection{Period-Two Orbits}

The next simplest type of orbits for the discrete map have period two and correspond to cycles of the flow.  A period-two orbit on $\Sigma$ must contain points for which $\theta=\mathcal{E}$ and $\theta=\mathcal{R}$ appear alternatively.  This can occur in a few different situations (see Fig.~\ref{fig:PeriodTwo}), which we summarize in the following subsections.  We include further details in the Appendix. Note that some of the period-two orbits correspond to higher-period orbits of the discrete dynamics on $\Sigma'$. For convenience, we represent such orbits on $\Sigma'$, with the understanding that when restricted to $\Sigma$ (i.e., when points with symbolic variable $\mathcal{S}$ are excluded), they all have period two.

\begin{figure}[htcp]
\centering
\includegraphics[scale=0.5]{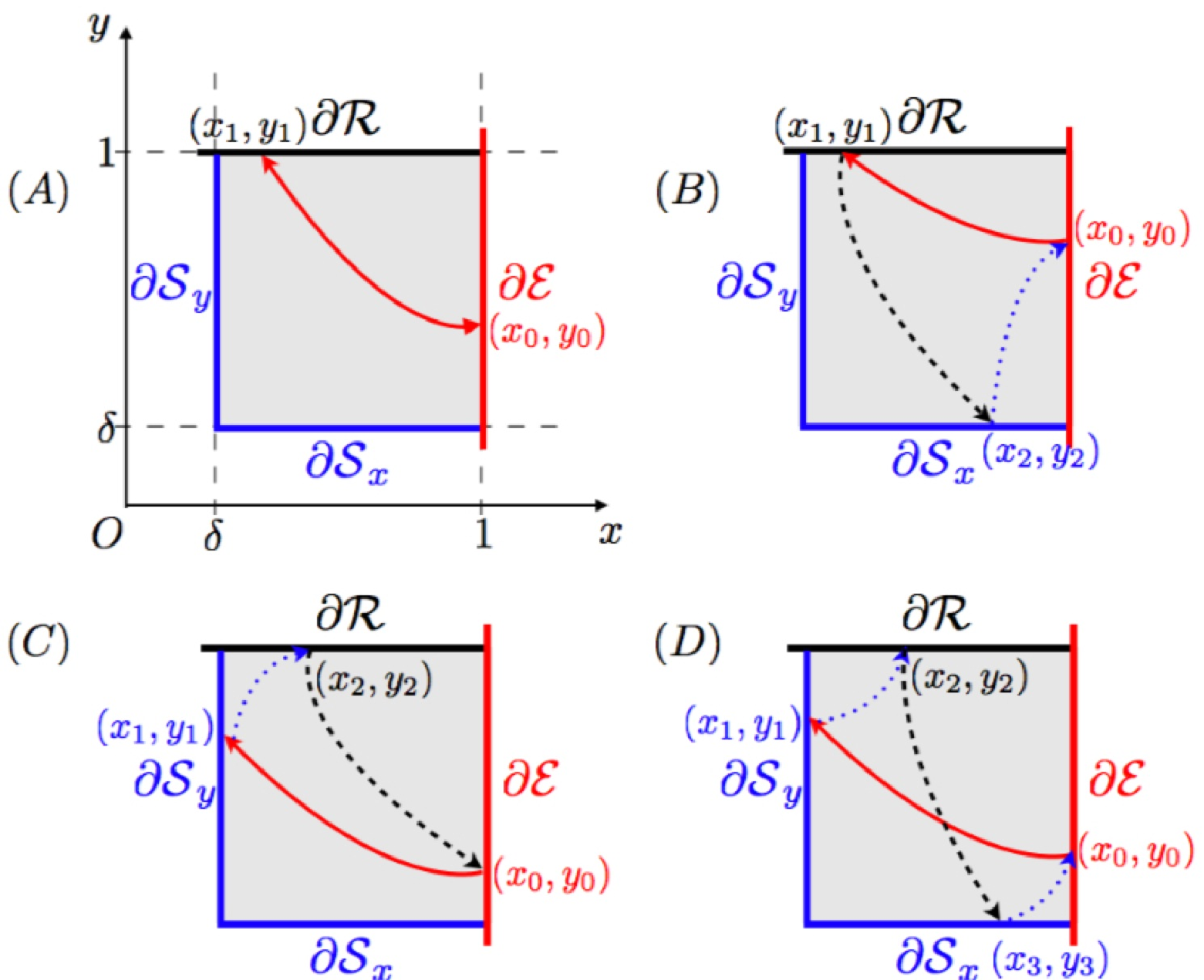}
\caption{(Color online) Illustration of all of the possible period-two orbits on $\Sigma$.
}
\label{fig:PeriodTwo}
\end{figure}


\subsubsection{Case A:~$(x_0,y_0;\mathcal{E})\rightarrow(x_1,y_1;\mathcal{R})\rightarrow(x_0,y_0;\mathcal{E})\rightarrow\dots$}

The existence of such an orbit requires the parameters to satisfy
\begin{equation}
	\frac{\alpha_2}{\alpha_1}\cdot\frac{\beta_2}{\beta_1} = 1\,.
\end{equation}
This implies that there are infinitely many stable but not asymptotically stable period-two orbits.  The initial value of $x$ is $x_0=1$ and the initial value of $y$ (called $y_0$, of course) is in the range
\begin{equation}
	\max\left(\delta,\delta^{\frac{\beta_1}{\alpha_2}}=\delta^{\frac{\beta_2}{\alpha_1}}\right)< y_0 <1\,.
\end{equation}
If $y_0$ is outside of this range, then one can see using numerical simulations that this orbit will necessarily contain a point in this range that 
\textcolor{red}{becomes}
a stable period-two orbit.


\subsubsection{Case B:~$(x_0,y_0;\mathcal{E})\rightarrow(x_1,y_1;\mathcal{R})\rightarrow(x_2,y_2;\mathcal{S}_x)\rightarrow(x_0,y_0;\mathcal{E})\rightarrow\dots$}

The parameters need to satisfy
\begin{equation}
	\frac{\alpha_2}{\alpha_1}\cdot\frac{\beta_2}{\beta_1}>1	
\end{equation}
or else any trajectory either converges to a fixed point or a stable period-two orbit (as discussed above).  It is also necessary that
\begin{equation}
	\frac{1}{\alpha_1}+\frac{1}{\alpha_2}\geq\frac{1}{\beta_1}+\frac{1}{\beta_2}
	\mbox{~if~}\beta_1<\alpha_2\,.
\end{equation}
There is only one period-two orbit of this type if the above two conditions hold.  This implies that $x_0=1$ and
\begin{equation}
	y_0=\delta^{\frac{1+\beta_1/\beta_2}{1+\alpha_2/\alpha_1}}\,.
\end{equation}
This periodic orbit is asymptotically stable if and only if
\begin{equation}\label{aa}
	\alpha_2<\alpha_1\,.
\end{equation}
That is, the orbit is asymptotically stable if and only if the rate at which a cow becomes sated while it eating is slower than the rate at which it becomes hungrier when it is not eating.


\subsubsection{Case C:~$(x_0,y_0;\mathcal{E})\rightarrow(x_1,y_1;\mathcal{S}_y)\rightarrow(x_2,y_2;\mathcal{R})\rightarrow(x_0,y_0;\mathcal{E})\rightarrow\dots$}

Again, we first need
\begin{equation}
	\frac{\alpha_2}{\alpha_1}\cdot\frac{\beta_2}{\beta_1}>1\,.
\end{equation}
Additionally,
\begin{equation}
	\frac{1}{\alpha_1}+\frac{1}{\alpha_2}<\frac{1}{\beta_1}+\frac{1}{\beta_2}
	\mbox{~and~}\beta_1<\alpha_2\,.
\end{equation}
There is also only one period-two orbit of this type; it has $x_0=1$ and
\begin{equation}
	y_0 = \delta^{\frac{1/\alpha_1+1/\alpha_2}{1/\beta_1+1/\beta_2}}\,.
\end{equation}
This orbit is asymptotically stable if and only if
\begin{equation}
	\beta_2<\beta_1\,.
\end{equation}
This case is analogous to case B, except that the roles of lying down and eating have been reversed.  Hence, this period-two orbit is asymptotically stable if and only if the rate at which a cow desires to get up when it is lying down is slower than the rate at which it increases its desire to lie down when it is not lying down.


\subsubsection{Case D:~$(x_0,y_0;\mathcal{E})\rightarrow(x_1,y_1;\mathcal{S}_y)\rightarrow(x_2,y_2;\mathcal{R})\rightarrow(x_3,y_3;\mathcal{S}_x)\rightarrow(x_0,y_0;\mathcal{E})\rightarrow\dots$}

The appearance of this orbit requires the following conditions to be satisfied:
\begin{eqnarray}
	\frac{\alpha_2}{\alpha_1}\cdot\frac{\beta_2}{\beta_1}&>&1\,,\nonumber\\
	\frac{1}{\alpha_1}+\frac{1}{\alpha_2}&=&\frac{1}{\beta_1}+\frac{1}{\beta_2}
	\mbox{~and~}\beta_1<\alpha_2\,.
\end{eqnarray}
There are infinitely many such orbits, which satisfy $x_0=1$ and 
\begin{equation}
	\delta<y_0<\delta^{\frac{\beta_1}{\alpha_2}}\,.
\end{equation}
All of these orbits are stable but not asymptotically stable.

\subsubsection{Summary}
We summarize the emergence of low-period orbits (up to period two) of $f:\Sigma\rightarrow\Sigma$ in different parameter ranges in Table~\ref{TableSummary}.

\begin{table}
\caption{Summary of low-period orbits (up to period two) and their stability of the single cow dynamics restricted to the {Poincar\'e} section $\Sigma$. All orbits except for the first one are period-two orbits on $\Sigma$. In the `Stability' column, we use `a.s' as an abbreviation for `asymptotically stable'.}
    \begin{tabular}{ | l | l | l | l | }
    \hline
    Parameters & Orbit & Condition on $y_0$ & Stability \\ \hline
    $\dfrac{\alpha_2}{\alpha_1}\cdot\dfrac{\beta_2}{\beta_1}<1$ & 
      $\{(1,1;\mathcal{E})\}$ & 
        none & 
          a.s \\ \hline
    $\dfrac{\alpha_2}{\alpha_1}\cdot\dfrac{\beta_2}{\beta_1}=1$ & 
      $\{(1,y_0;\mathcal{E}),(y_{0}^{\frac{\alpha_2}{\beta_1}},1;\mathcal{R})\}$&
        $\max{(\delta,\delta^{\frac{\beta_1}{\alpha_2}})}<y_0<1$&
    	  stable \\ \hline
    $\dfrac{\alpha_2}{\alpha_1}\cdot\dfrac{\beta_2}{\beta_1}>1, \alpha_2<\beta_1$ & 
      $\{(1,y_0;\mathcal{E}),(y_{0}^{\frac{\alpha_2}{\beta_1}},1;\mathcal{R})\}$ &
        $y_0=\delta^{\frac{1+\frac{\beta_1}{\beta_2}}{1+\frac{\alpha_2}{\alpha_1}}}$& 
          a.s iff $\alpha_2<\alpha_1$\\ \hline
    $\begin{cases}\dfrac{\alpha_2}{\alpha_1}\cdot\dfrac{\beta_2}{\beta_1}>1, \alpha_2>\beta_1;\\
    	\dfrac{1}{\alpha_1}+\dfrac{1}{\alpha_2}\geq\dfrac{1}{\beta_1}+\dfrac{1}{\beta_2}\end{cases}$&
    		$\{(1,y_0;\mathcal{E}),(y_{0}^{\frac{\alpha_2}{\beta_1}},1;\mathcal{R})\}$ & 
	$y_0=\delta^{\frac{1+\frac{\beta_1}{\beta_2}}{1+\frac{\alpha_2}{\alpha_1}}}$& 
	a.s iff $\alpha_2<\alpha_1$\\ \hline
    $\begin{cases}\dfrac{\alpha_2}{\alpha_1}\cdot\dfrac{\beta_2}{\beta_1}>1, \alpha_2>\beta_1;\\
    	\dfrac{1}{\alpha_1}+\dfrac{1}{\alpha_2}<\dfrac{1}{\beta_1}+\dfrac{1}{\beta_2}\end{cases}$&
    	  $\{(1,y_0;\mathcal{E}),(\delta,\delta^{-\frac{\beta_1}{\alpha_2}}y_0;\mathcal{R})\}$ & 
	    $y_0=\delta^{\frac{\frac{1}{\alpha_1}+\frac{1}{\alpha_2}}{\frac{1}{\beta_1}+\frac{1}{\beta_2}}}$& 
	      a.s iff $\beta_2<\beta_1$\\ \hline
    $\begin{cases}\dfrac{\alpha_2}{\alpha_1}\cdot\dfrac{\beta_2}{\beta_1}>1, \alpha_2>\beta_1;\\
    	\dfrac{1}{\alpha_1}+\dfrac{1}{\alpha_2}=\dfrac{1}{\beta_1}+\dfrac{1}{\beta_2}\end{cases}$&
    	  $\{(1,y_0;\mathcal{E}),(\delta^{1+\frac{\alpha_1}{\alpha_2}}y_{0}^{-\frac{\alpha_1}{\beta_1}},1;\mathcal{R})\}$ & 
	    $\delta<y_0<\delta^{\frac{\beta_1}{\alpha_2}}$ & 
	      stable \\ \hline
     \end{tabular}
     \label{TableSummary}
\end{table}


\subsection{Grazing Bifurcations}

We remark that the single cow equations cannot exhibit grazing bifurcations.\footnote{In the theory of piecewise smooth dynamical systems, a \emph{grazing bifurcation} is said to occur when a limit cycle of a flow becomes tangent to a discontinuity boundary  \cite{Bernardo_2007,scholarpiecewise}.}


\subsection{Higher-Period Orbits and Bifurcation Diagram}

Although one could proceed to analyze more complicated orbits, this is not the main topic of this paper.  Instead, we simply illustrate the existence of more complicated (possibly chaotic) orbits through a bifurcation diagram (see Fig.~\ref{fig:Bifurcation}) by simulation with varying one of the parameters. This parameter, which we choose to be $\alpha_2$, seems to be transverse to the unfolding of the bifurcation and reveals rich dynamics in our model.

\begin{figure}[htcp]
\centering
\includegraphics[scale=0.7]{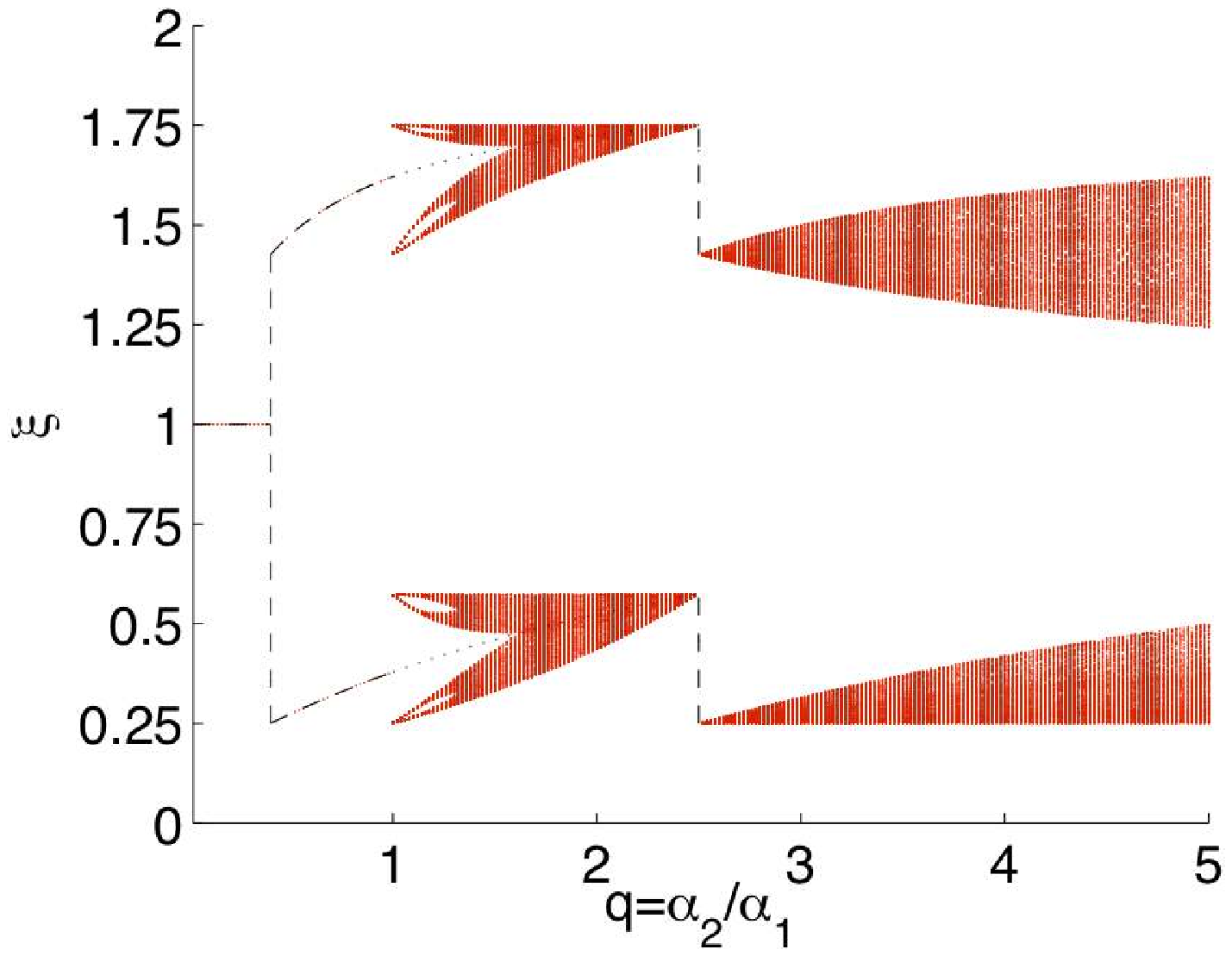}
\includegraphics[scale=0.7]{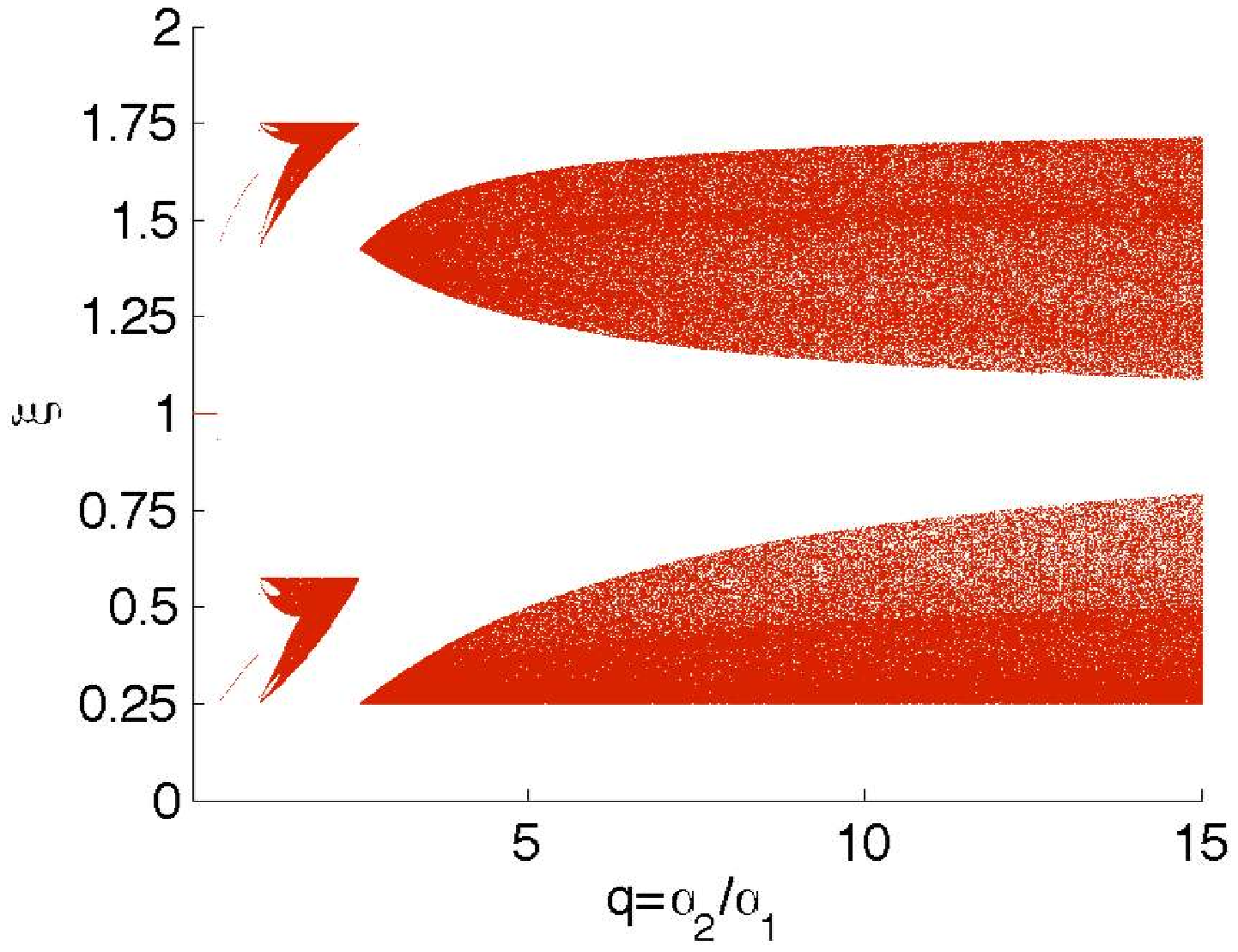}
\caption{(Color online) Bifurcation diagram for the discrete dynamics $f$ on $\Sigma$.  We fix the parameter values $\alpha_1=0.05$, $\beta_1=0.05$, $\beta_2=0.125$, and $\delta=0.25$. 
The vertical axis is $\xi=y+(1-x)$, corresponding to the points $(x,y)$ on $\Sigma$.
In the top panel, we show the diagram for which $q\equiv\frac{\alpha_2}{\alpha_1}$ ranges from $0$ to $5$; dashed and dotted curves give theoretical results, which we summarize in Table \ref{TableSummary}. In the bottom panel, we show the diagram for $q$ from $0$ to $15$. If we further increase $q$, the two large finger-like bands on the right of the diagram retain their shape and become progressively closer.  Numerical simulations suggest that the distance between them tends to $0$ as $q\rightarrow\infty$.  
}
\label{fig:Bifurcation}
\end{figure}

For a wide range of parameters, there seems to always be a dense subset (for a fixed set of parameters) of the domain that attracts ``typical" (in the sense of nonzero measure) initial conditions.  We show one of these (likely chaotic) orbits in Fig.~\ref{fig:TypicalOrbit}.  We connect the dots by straight lines in order to illustrate the end points of the flow touching the boundaries, although the actual trajectories between points on the boundaries are convex curves.  We remark that one can think of the discrete dynamics on $\Sigma'$ as a billiard-like problem (see Refs.~\cite{predrag,sinai} and references therein for discussions of billiards) with nontrivial bouncing rules on the boundary and nonlinear potentials that determine the trajectories of particles between collisions with the boundary.

\begin{figure}[htcp]
\centering
\includegraphics[scale=0.7]{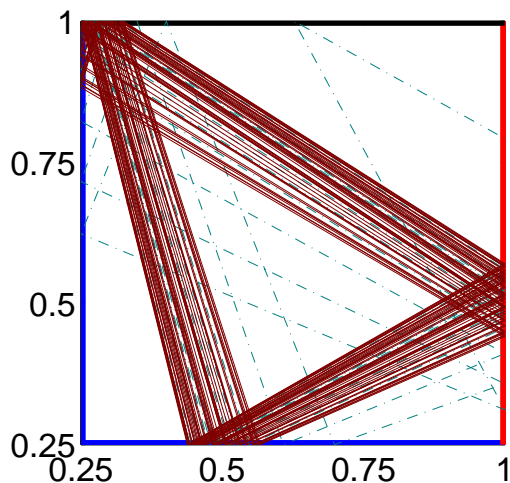}
\caption{(Color online) A typical discrete orbit (thin solid lines) on $\Sigma$ for the parameters $\alpha_1=0.05$, $\alpha_2 =0.1$, $\beta_1=0.05$, $\beta_2=0.125$, and $\delta=0.25$. We depict the case corresponding to $q=2$ in Fig.~\ref{fig:Bifurcation}.  The dashed lines show transient dynamics.    We highlight the boundaries using thick solid lines.  For aesthetic reasons, we join successive points on $\Sigma$ with straight lines, and we note that the actual flow that connects these points are piecewise convex curves.
}
\label{fig:TypicalOrbit}
\end{figure}


\section{Coupled Cows and Synchronization} \label{four}

As we discussed in the introduction, there are many biological benefits to achieving synchronized eating and lying down in cattle.  We are thus motivated to construct \emph{herd equations} that describe interacting cows by coupling the single cow equations 
(\ref{eq:singlecowE}--\ref{eq:singlecowSwitch}).
We make specific choices motivated by biology and simplicity, though it  is of course important to consider both more complicated choices and alternative forms of coupling.  Our goal is to highlight just one possible form of the interactions in detail, but we hope that our work will serve as a springboard for rumination of some of the alternatives that we will mention briefly in Section \ref{five}. 

In this section, we numerically investigate the effect of coupling in a system of a few cows. For the purpose of simplifying the exposition of the equations, we use indicator functions defined on the set $\{\mathcal{E},\mathcal{R},\mathcal{S}\}$:
\begin{equation}
	\chi_{\psi}(\theta)\equiv
	\begin{cases}
	1\,, & \mbox{~if~}\theta=\psi\,,\\
	0\,, & \mbox{~otherwise~}\,.
	\end{cases}
\end{equation}

The single cow equation in between state transitions (\ref{eq:singlecow}) can then be written as
\begin{equation}
	\begin{cases}
		\dot{x} = \alpha(\theta)x\,,\\
		\dot{y} = \beta(\theta)y\,,
	\end{cases}
\end{equation}
where we have defined functions
\begin{equation}
	\begin{cases}
		\alpha(\theta)=-\alpha_2\chi_{\mathcal{E}}(\theta)+\alpha_1\chi_{\mathcal{R}}(\theta)+\alpha_1\chi_{\mathcal{S}}(\theta)\,,\\
		\beta(s)=\beta_1\chi_{\mathcal{E}}(\theta)-\beta_2\chi_{\mathcal{R}}(\theta)+\beta_1\chi_{\mathcal{S}}(\theta)\,.
	\end{cases}
\end{equation}


\subsection{Coupling Scheme}

There are numerous possible ways to model the coupling between cows. We have chosen one based on the hypothesis that a cow feels hungrier when it notices the other cows eating and feels a greater desire to lie down when it notices other cows lying down.  (We briefly discuss other possibilities in Section \ref{five}.)  This provides a coupling that does not have a spatial component, in contrast to the agent-based approach of Ref.~\cite{franz}.  We therefore assume implicitly that space is unlimited, so we are considering cows to be in a field rather than in a pen.  We suppose that the herd consists of $n$ cows and use $i$ to represent the $i$-th cow in the herd.  This yields herd equations given by
\begin{equation}\label{eq:HerdEquations}
	\begin{cases}
		\dot{x_i}=\big[\alpha^{(i)}(s_i) + \frac{\sigma_{x}}{k_i}\sum_{j=1}^{n}a_{ij}\chi_{\mathcal{E}}(s_j)\big]x_i\,,\\
		\dot{y_i}=\big[\beta^{(i)}(s_i) + \frac{\sigma_{y}}{k_i}\sum_{j=1}^{n}a_{ij}\chi_{\mathcal{R}}(s_j)\big]y_i\,,
	\end{cases}
\end{equation}
with switching condition according to Eq.~(\ref{eq:singlecowSwitch}) for each individual cow.
The second terms in both equations give the coupling terms of this system.  The matrix $A=[a_{ij}]_{n\times{n}}$ is a {time-dependent} adjacency matrix that represents the network of cows.  Its components are given by
\begin{equation}
	a_{ij}(t) = 
	\begin{cases}
		1 & \mbox{~if the $i$-th cow interacts with the $j$-th cow at time $t$}\,,\\
		0 & \mbox{~if the $i$-th cow does not interact with the $j$-th cow at time $t$}\,.
	\end{cases}
\end{equation}
Additionally, $k_i=\sum_{j=1}^{n}A_{ij}$ is the {degree} of node $i$ (i.e., the number of cows to which it is connected), and the coupling strengths $\sigma_x$ and $\sigma_y$ are non-negative (and usually positive) real numbers corresponding to the strength of coupling.  This is designed to emphasize that animal interaction strengths consider proximity to neighboring animals.

It is important to note that in the case where $A$ is time-independent, the dynamics governing the network of interacting cows only changes when at least one of the individual cows changes its state $\theta_i$. In practice, we can solve analytically for the flows in between such transitions (because they are piecewise affine differential equations) instead of performing numerical integration in the whole time interval, which might cause numerical instability when the number of transitions becomes large.


\subsection{Measuring Synchrony}


We also need a measure for the synchrony between cows. For each cow $i$, let $\tau^{(i)}$ and $\kappa^{(i)}$ be vectors such that
\begin{equation}
	\begin{cases}
	\tau^{(i)}_{k}\equiv\mbox{The $k$-th time at which the $i$-th cow switches its state to $\mathcal{E}$}\,,\\
	\kappa^{(i)}_{k}\equiv\mbox{The $k$-th time at which the $i$-th cow switches its state to $\mathcal{R}$}\,.
	\end{cases}
\end{equation}

Given pairs of vectors $\tau^{(i)}$ and $\tau^{(j)}$ of the same length, the ``eating" synchrony between cows $i$ and $j$ is measured by
\begin{equation}
	\Delta^{\mathcal{E}}_{ij}\equiv\langle|\tau^{(i)}-\tau^{(j)}|\rangle\,
	= \frac{1}{K}\sum_{k=1}^{K}{|\tau^{(i)}_{k}-\tau^{(j)}_{k}|},
\end{equation}
where $\langle \cdot \rangle$ denotes time-averaging.
In general, the vectors $\tau^{(i)}$ and $\tau^{(j)}$ are of different lengths, so we truncate and shift one of them to match up with the other in such a way that it gives approximately the minimal $\Delta^{\mathcal{E}}_{ij}$ as defined above.

Similarly, we define the ``lying" synchrony between cows $i$ and $j$ by
\begin{equation}
	\Delta^{\mathcal{R}}_{ij}\equiv\langle|\kappa^{(i)}-\kappa^{(j)}|\rangle\,.
\end{equation}

For $n$ cows, the group ``eating" and ``lying" synchrony are then measured by averaging over all of the synchrony between individual pairs:
\begin{equation}\label{eq:SyncMeasureER}
	\begin{cases}
	\Delta^{\mathcal{E}}\equiv\langle\Delta^{\mathcal{E}}_{ij}\rangle\, = \frac{1}{n^2}\sum_{i,j}\Delta^{\mathcal{E}}_{ij},\\
	\Delta^{\mathcal{R}}\equiv\langle\Delta^{\mathcal{R}}_{ij}\rangle\, = \frac{1}{n^2}\sum_{i,j}\Delta^{\mathcal{R}}_{ij},
	\end{cases}
\end{equation}
and the aggregate synchrony can then be measured via
\begin{equation}
	\Delta\equiv\Delta^{\mathcal{E}}+\Delta^{\mathcal{R}}.
\end{equation}
There are, of course, other possible measures of synchrony that one could employ.
 For example, in his agent-based study, Franz \cite{franz} considered kappa statistics, an order parameter adapted from the usual one used in the Kuramoto model, and a direct count of how often all cows are lying down \cite{nielsen97,faerevik08}.


\subsection{Numerical Exploration of Herd Synchrony}

With the tools described above, we are now ready to show some examples of synchronization of cows. We will start with a system consisting of only  two cows and then consider herds with more than two cows.


\subsubsection{Two Coupled Cows}

We first examine how the coupling strength affects the extent of synchronization. Assume that the two cows have individual dynamics that are specified by nearly identical parameters: 
\begin{eqnarray}\label{eq:CoupledTwoCows}
	\alpha_{1}^{(1,2)}&=&0.05\pm{\epsilon}\,, \quad \alpha_{2}^{(1,2)}=0.1\pm\epsilon\,,\nonumber\\
	\beta_{1}^{(1,2)}&=&0.05\pm\epsilon\,, \quad \beta_{2}^{(1,2)}=0.125\pm\epsilon\,,\nonumber\\
	\delta&=&0.25\,.
\end{eqnarray}
We show simulation results in Fig.~\ref{fig:TwoCows} and Fig.~\ref{fig:TwoCowsSync} to illustrate the dependence of synchrony both on the parameter mismatch $\epsilon$ and on the coupling strength $\sigma_{x}, \sigma_{y}$.

\begin{figure}[htcp]
\includegraphics[scale=0.43]{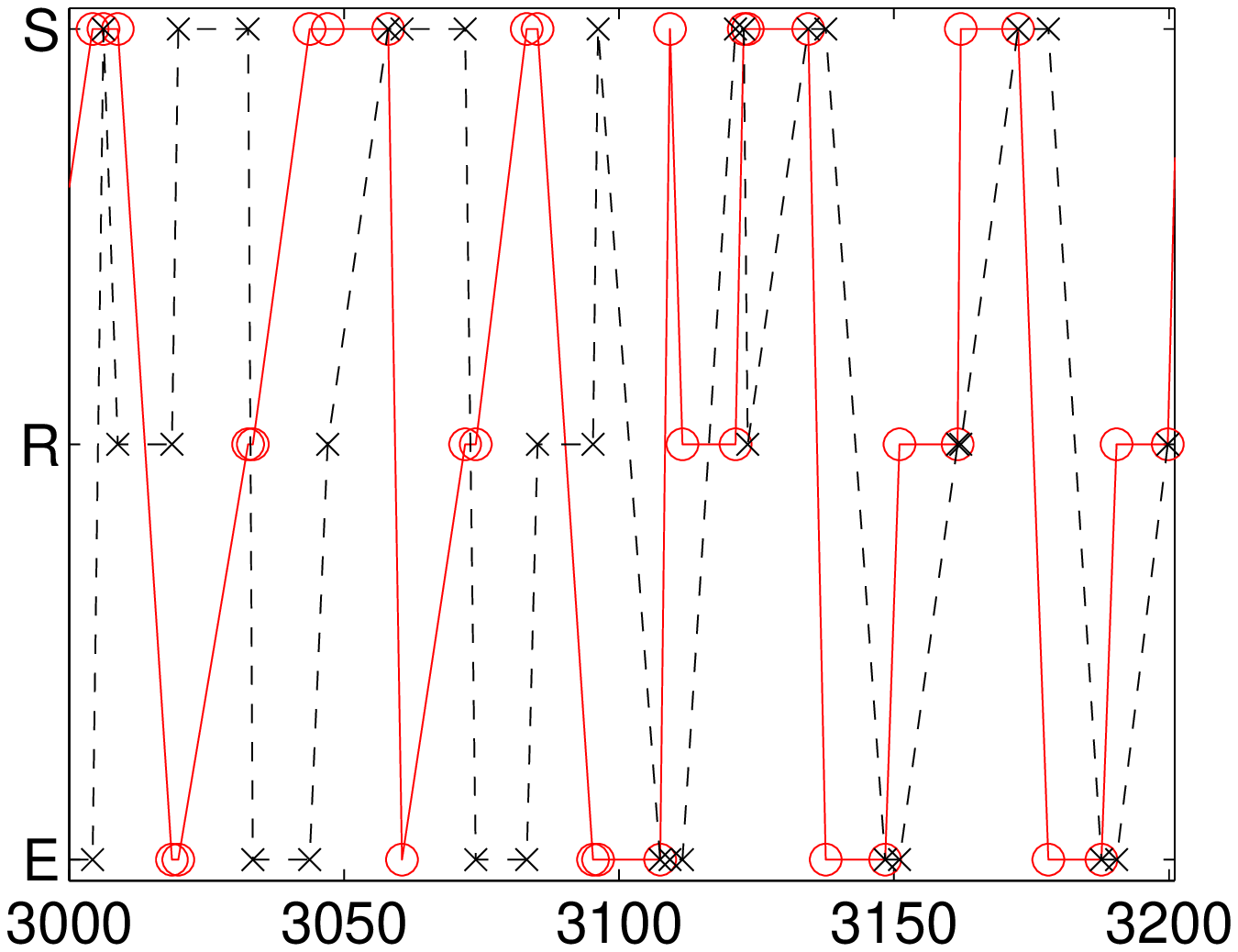}
\includegraphics[scale=0.43]{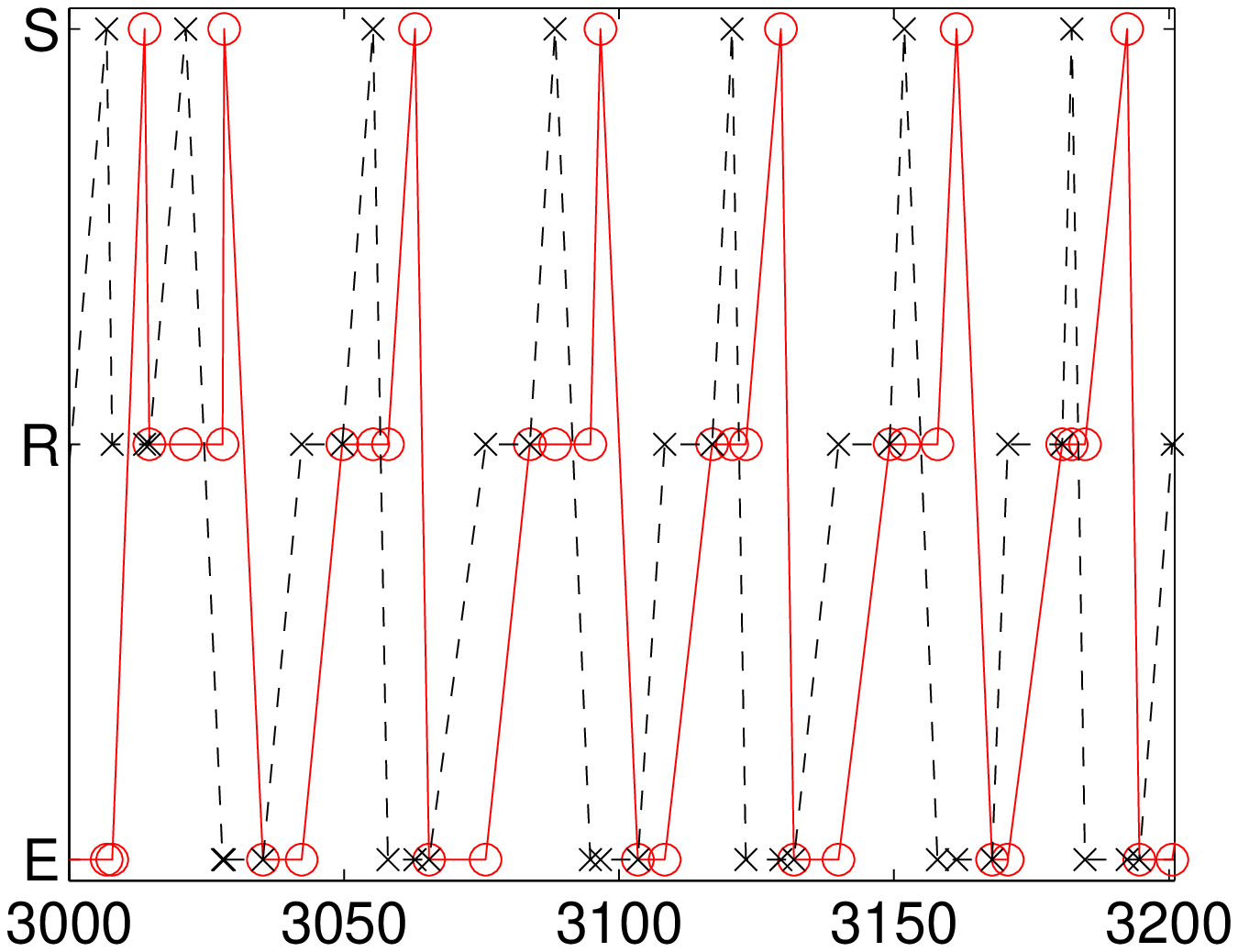}
\caption{(Color online) Typical time series of the state variables $\theta_{1,2}$ for different coupling strengths.  The system of equations is described by Eq.~(\ref{eq:HerdEquations}), and the parameter values are given in 
Eq.~(\ref{eq:CoupledTwoCows}).  
The parameter mismatch between the two cows is $\epsilon=10^{-3}$.  The horizontal axis is time $t$.  The left panel shows the transition of states $\theta_1$ (red circles connected by $`-'$) and $\theta_2$ (black crosses connected by $`--'$) of a typical time series with the coupling strengths $\sigma_x=\sigma_y=0$ (i.e., when there is no coupling).  The right panel shows a similar plot with the coupling strengths $\sigma_x=\sigma_y=0.045$.
}
\label{fig:TwoCows}
\end{figure}

\begin{figure}[htcp]
\includegraphics[scale=0.43]{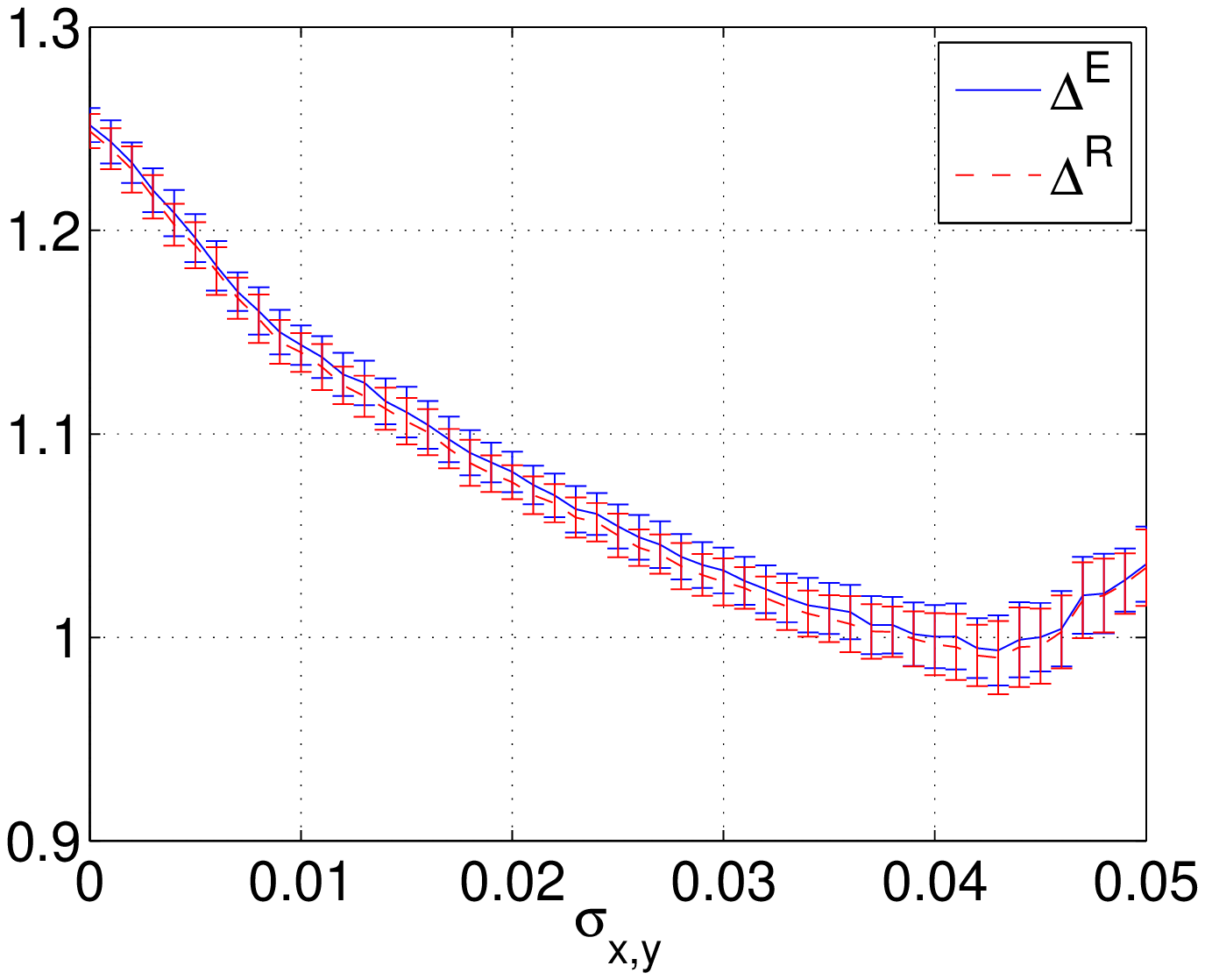}
\includegraphics[scale=0.43]{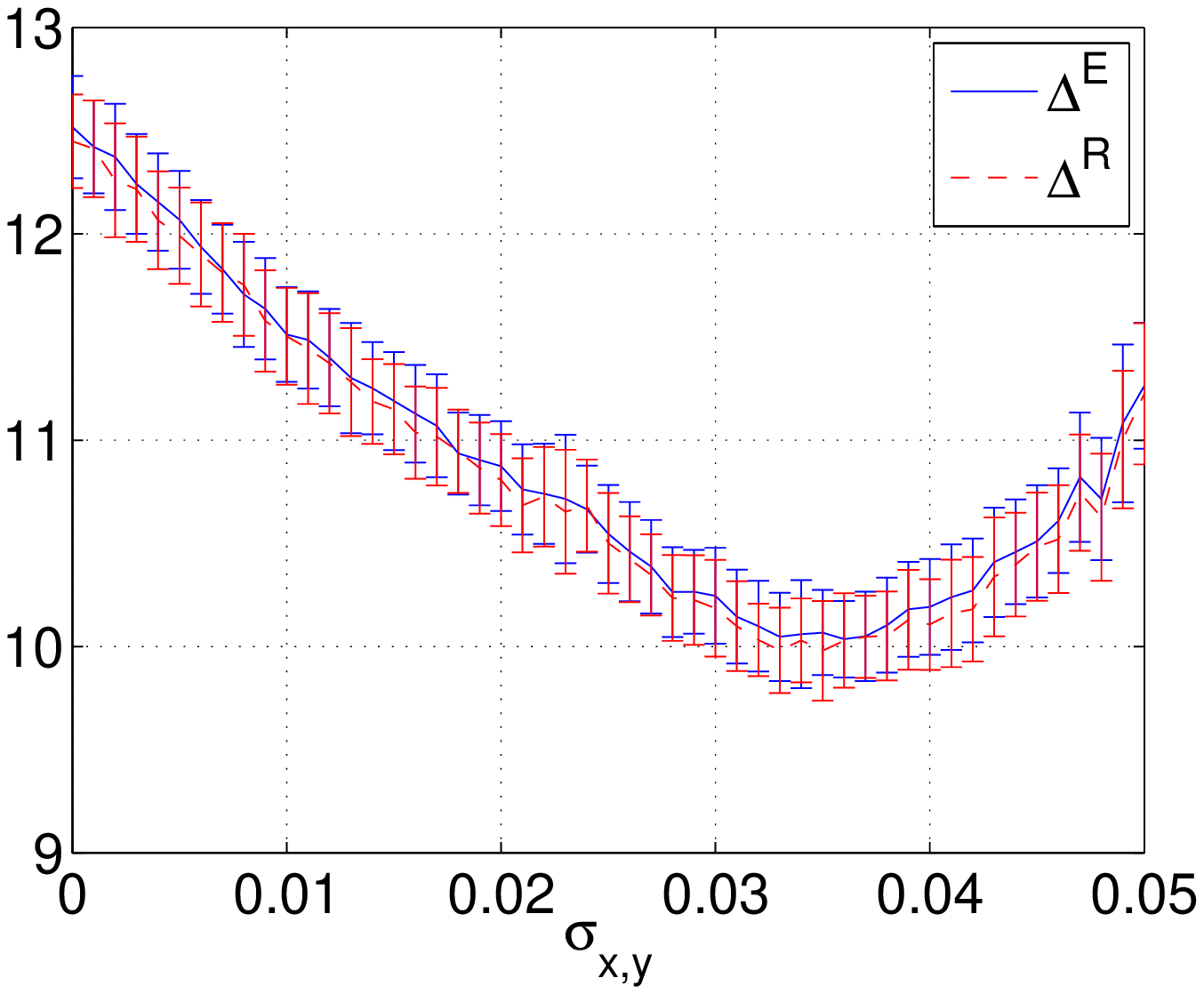}
\caption{(Color online) Dependence of synchrony on coupling strength. The system of equations and parameters are specifed by Eq.~(\ref{eq:HerdEquations}) and Eq.~(\ref{eq:CoupledTwoCows}), respectively.  In the left panel, we illustrate the synchronization error, which we measure using Eq.~(\ref{eq:SyncMeasureER}), for different coupling strengths $\sigma_{x,y}$  for two coupled cows whose parameter mismatch is $\epsilon=10^{-3}$.  In the right panel, we show the synchronization error for parameter mismatch $\epsilon=10^{-2}$.  
We obtain each curve (for a fixed $\epsilon$) 
by averaging over 50 runs.  Each run is an independent realization of the herd equations starting from an initial condition chosen uniformly at random.
Vertical error bars indicate one standard deviation from the mean.
}
\label{fig:TwoCowsSync}
\end{figure}

These pictures suggests that our measure of synchrony is reasonable 
for such a system.  The greater the difference between the two cows, the harder it is for them to achieve synchrony.  However, the dependence of synchrony is not necessarily monotonically dependent on the coupling strength. An increase in the coupling strength at the beginning does improve synchrony, but there is a point beyond which larger coupling can in fact lead to lower synchrony.



\subsubsection{Network of Coupled Cows}

In this subsection, we show numerical results on synchronization among a few cows. In all examples, we consider a herd of $n=10$ cows, where each individual has parameter values slightly perturbed from $\alpha_1=0.05$, $\alpha_2=0.1$, $\beta_1=0.05$, and $\beta_2=0.125$.  Additionally, we note that $10^{-3}$ is the maximum difference in each parameter value relative to the average parameter among all individuals.  We can couple these cows using different network architectures---for example, a circular lattice and a star graph (see Fig.~\ref{fig:CowNetworks}).  We use these networks only as illustrative examples, as one can of course perform similar investigations with any other network architecture.

\begin{figure}[htcp]
\centering
\includegraphics[scale=0.30]{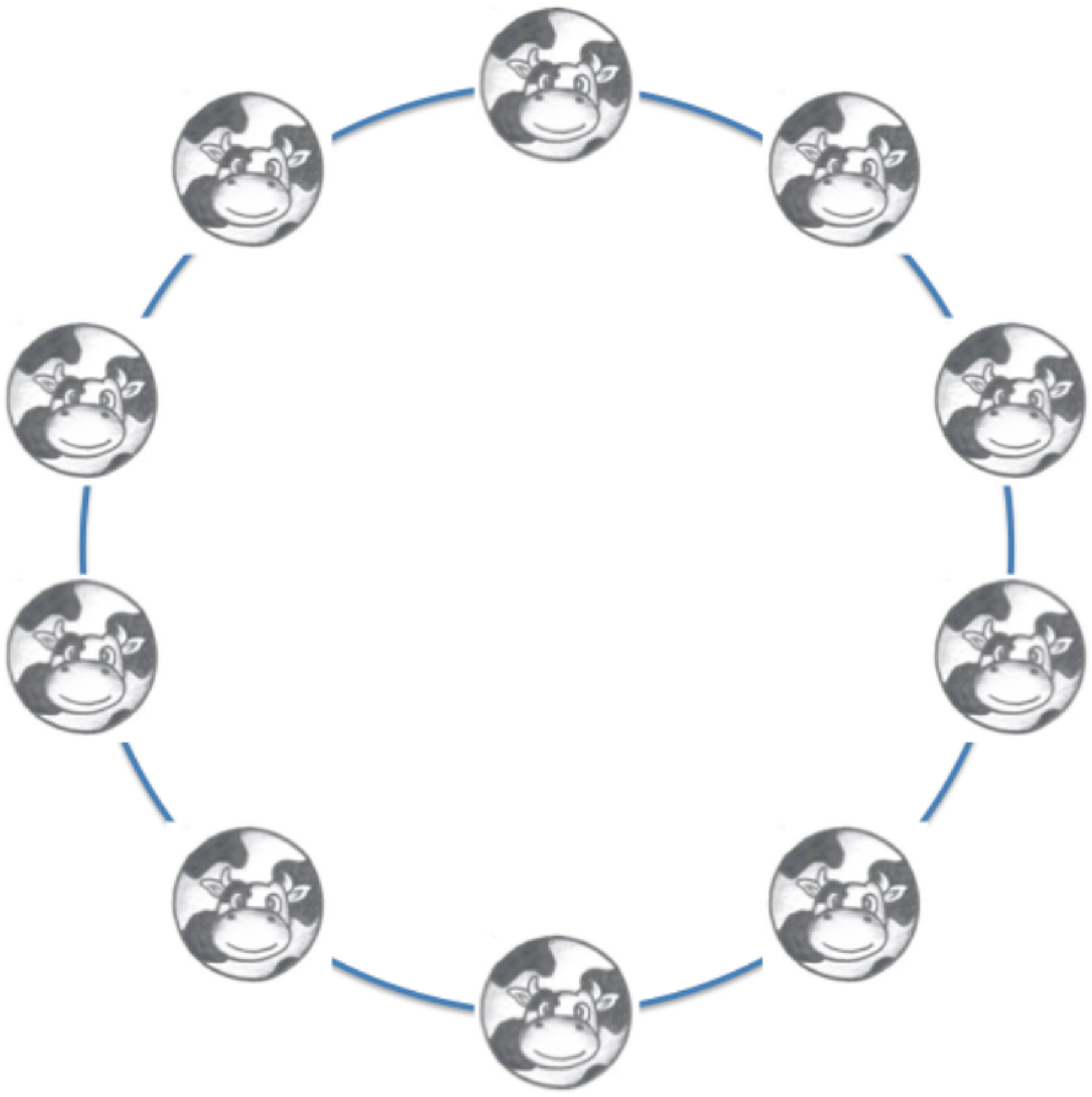}
\includegraphics[scale=0.30]{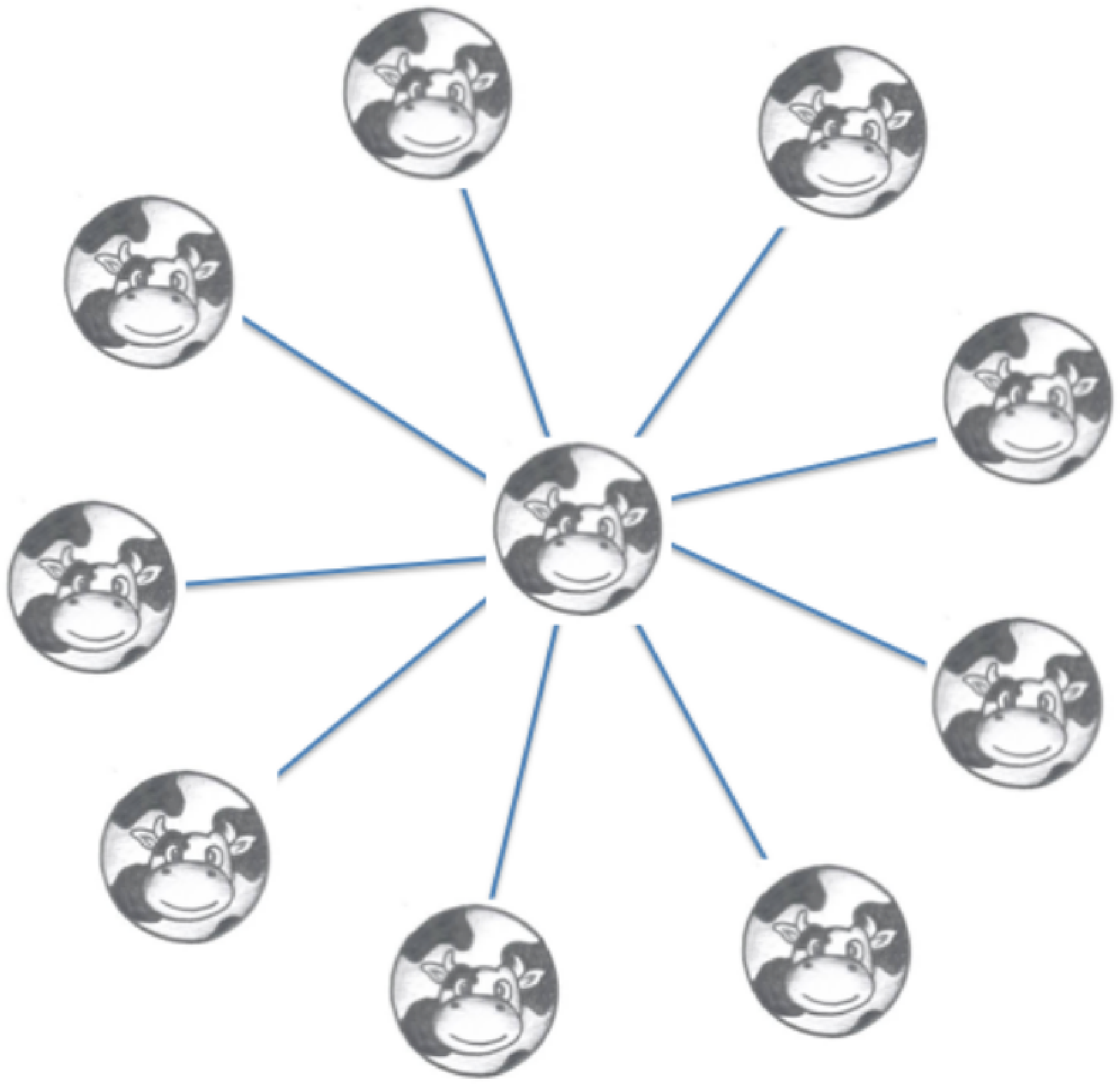}
\caption{(Color online) Example network architectures for coupled cows: (left) Circular lattice with $10$ nodes and (right) star graph with $10$ nodes.  [The spherical cow image was created for this paper by Yulian Ng and used with her permission.]  
}
\label{fig:CowNetworks}
\end{figure}

In Fig.~\ref{fig:TenCows}, we show the state transitions of the ten cows during a small time interval.  We consider fixed coupling strengths $\sigma_x=\sigma_y=0.05$ for each of the two network architectures.

\begin{figure}[htcp]
\includegraphics[scale=0.43]{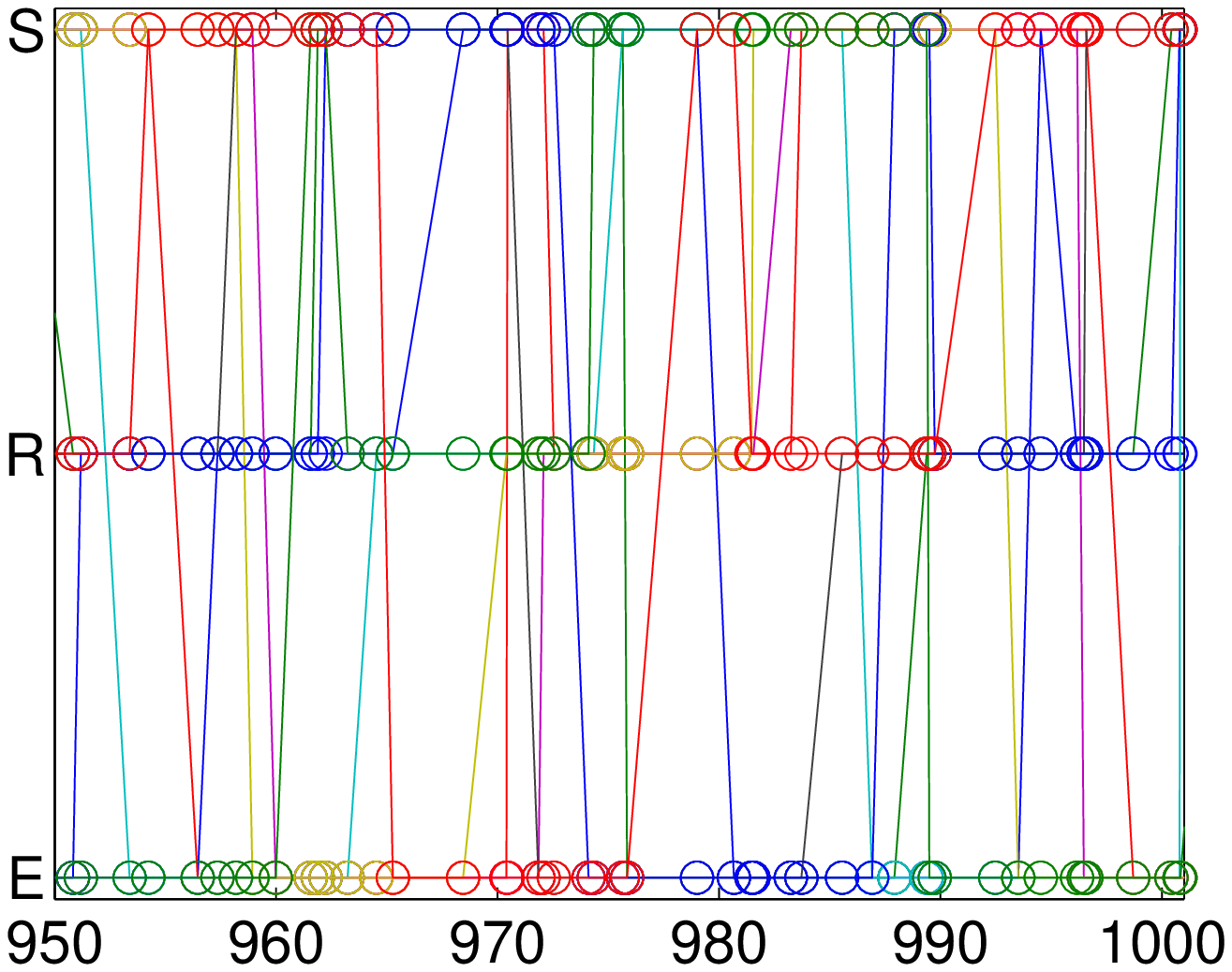}
\includegraphics[scale=0.43]{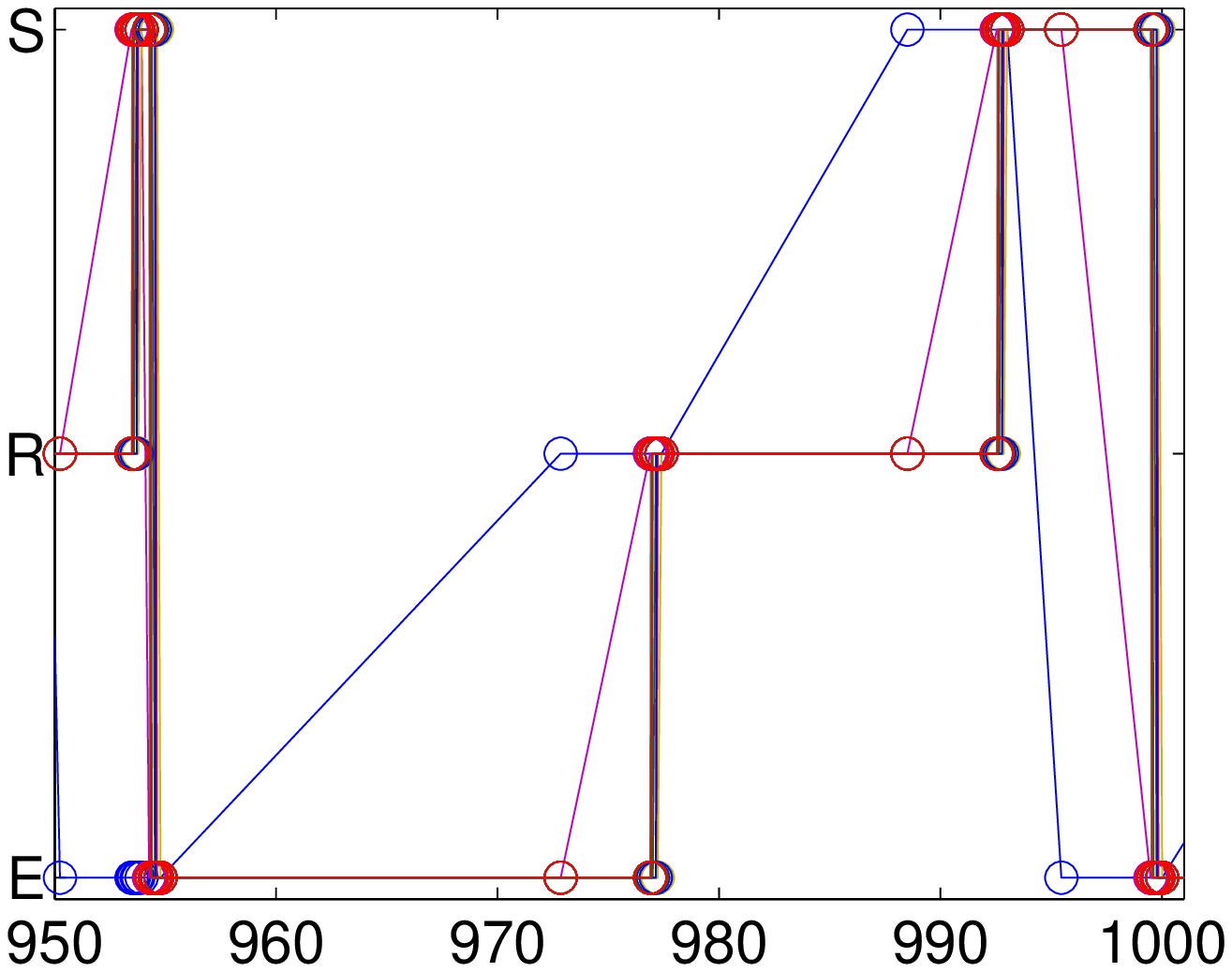}
\caption{(Color online) Typical state transitions for coupled cows in (left) a circular lattice and (right) a star graph with fixed coupling strengths $\sigma_x=\sigma_y=0.05$.  We plot (artificial) straight lines to help visualize transitions between states (which are represented by open circles, with different colors representing different cows).  The horizontal axis is time. Some of the curves overlap (so that fewer than 10 colors are visible) due to the partial synchrony between individual cows.
}
\label{fig:TenCows}
\end{figure}

In Fig.~\ref{fig:TenCowsSync}, we illustrate the dependence of synchrony on different coupling strengths for the two network configurations. Interestingly, when the coupling strength is increased, the cows tend to synchronize less when they are coupled via a circular lattice, whereas synchrony is improved if they are coupled via a star graph. We have also tested numerically other network configurations, such as circular lattices with more than just nearest-neighbor connections and (Erd\"os-Reny\'i) random graphs, and the resulting curves are qualitatively similar to the one shown in the left panel of Fig.~\ref{fig:TenCowsSync}.  It would be interesting to study what network architectures can lead to good synchrony beyond the star graph, which is an idealized example. A heuristic reason that synchrony can decrease when coupling is increased in our herd model (and, more generally, in piecewise smooth dynamical systems) is that decreasing the difference in the observable variables $(x,y)$ through coupling does not necessarily reduce the difference in the hidden state variable $\theta$, and the effect of coupling might well be the opposite from what one would naively anticipate, as we have observed using the circular lattice structure.  Moreover, recent work in other contexts has illustrated that increasing the number of connections in a network can sometimes lead to less synchrony \cite{nish10}.  Although synchronization has been studied extensively for smooth dynamical systems
~\cite{Arenas_PHYREP08, Belykh_PHYSICAD06, nish10, Pecora_PRL98, pikovsky, Senthilkumar_PRE09, Stilwell_SIADS06, sync, Sun_SIADS09, Sun_EPL09} 
and the mechanisms that promote synchrony in such situations are relatively well understood, little is known about networks of coupled piecewise smooth dynamical systems.  It would be interesting to study the influence of network architecture on synchrony for models other than smooth dynamical systems (such as the herd model considered in this paper), which might prove to be important in studying the behavior of interacting animals.


\begin{figure}[htcp]
\includegraphics[scale=0.43]{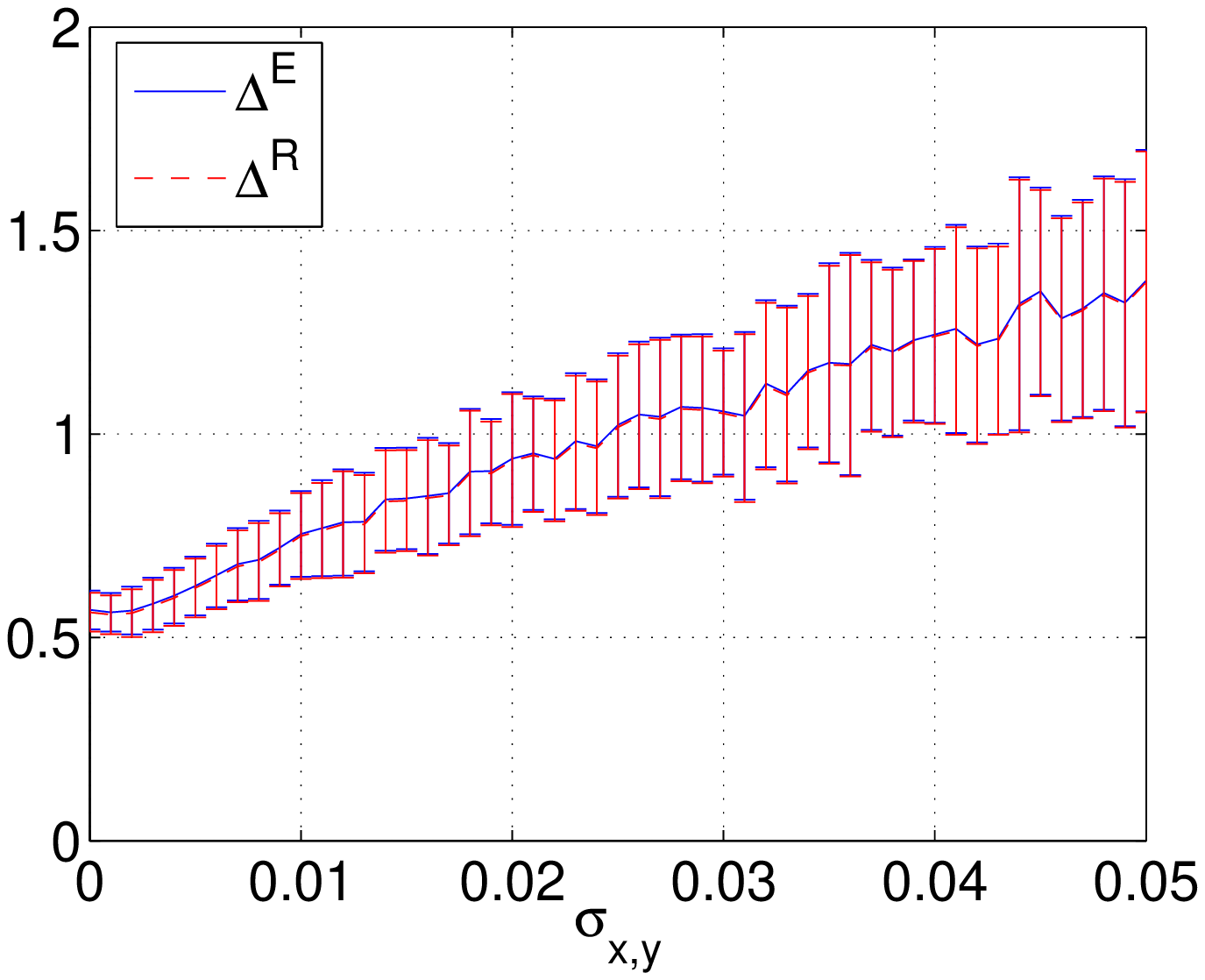}
\includegraphics[scale=0.43]{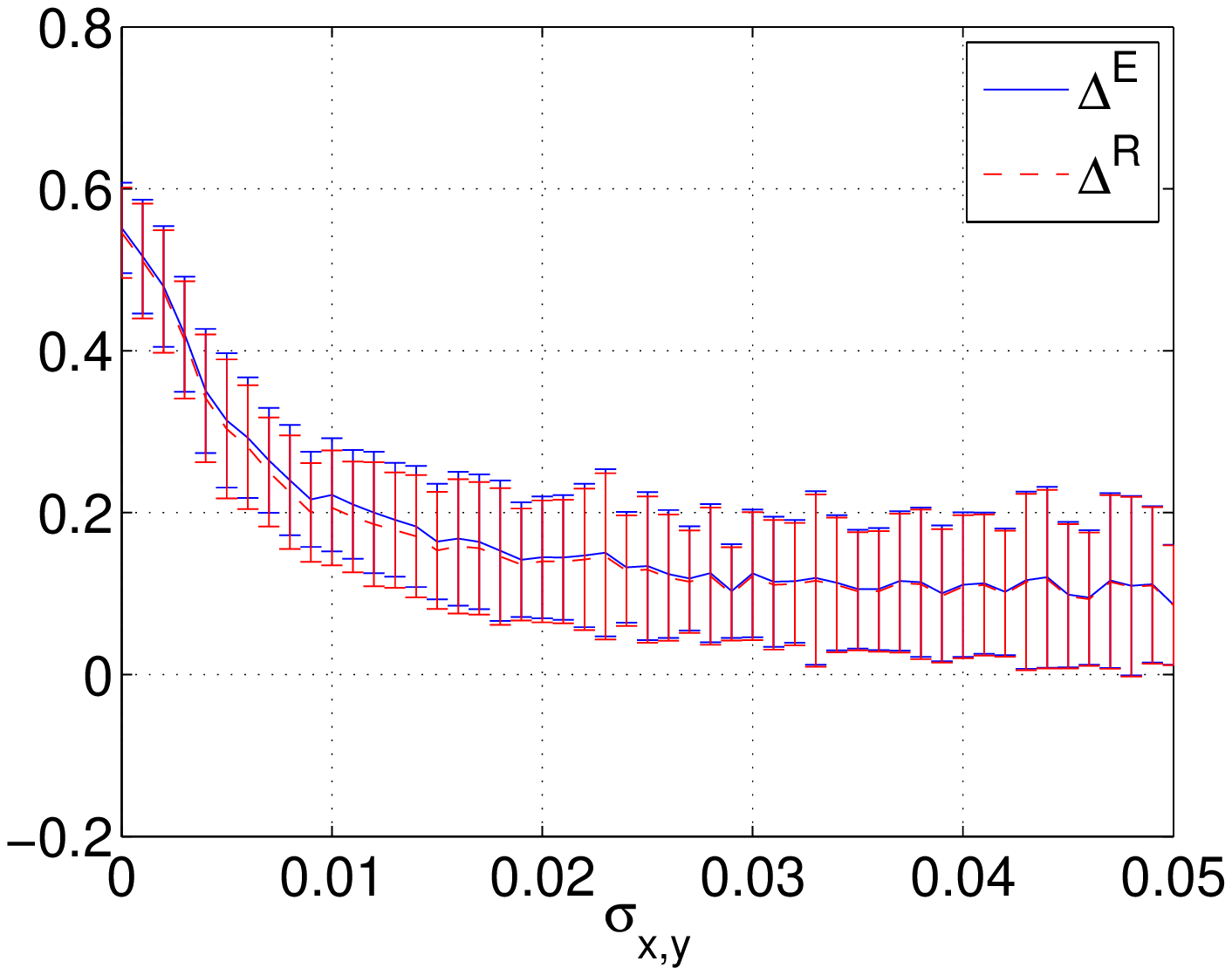}
\caption{(Color online) Synchrony measure versus coupling strengths in the (left) circular lattice and (right) star graph.
}
\label{fig:TenCowsSync}
\end{figure}


\section{Discussion} \label{five}

We have only scratched the surface concerning the modeling of herd synchrony in cattle. 

We considered each cow as an oscillator, which we modeled as a piecewise affine dynamical system.  Our single cow model had interesting mathematical properties, which we discussed in detail.  Monotonic dynamics within each state was the most important detail, and we chose affine monotonic dynamics to make the analysis as tractable as possible.  

We illustrated herd dynamics through specific coupling choices between cows.  We assumed that the herd is in a field rather than a pen and, in particular, ignored the presence of spatial constraints.  We considered cows that become hungrier when they notice others eating and a greater desire to lie down when they notice others lying down, but numerous other choices would also be interesting to study.  For example, the relative importances of the two aforementioned types of positive coupling can be varied systematically, and the specific functional forms of coupling can also, of course, be different.  Additionally, it is not clear whether cow synchrony arises from such an active mechanism or whether it can arise from more passive forms of coupling.  In particular, this could entail the incorporation of spatial effects, such as limited eating and bedding areas and the competition of cows for such resources.  The inherent oscillations of individual cows can lead to synchronization even with almost no interactions between individuals \cite{tom}.  Synchrony can potentially emerge even if the only interaction between cows occurs when one steps on another one, so such a minimalist but biologically meaningful mechanism (in which the cows need not even notice whether another cow is feeding or resting) would be interesting to test against more complicated forms of interaction.  It would also be interesting to use real observations of cattle to compare the synchronization properties in limited space versus ``unlimited" space (i.e., pens versus open fields), and such experiments are currently in progress.  

One could examine spatial effects in the oscillator model of cows by considering more realistic network architectures.  Such networks could either come from experimental data (which has not yet been gathered) of which cows come into contact with each other or using structures that respect the fact that fields and pens are planar regions.  It would also be interesting to consider different network structures from an abstract perspective in order to test observations such as the different dynamics with the star graph (which has one high-degree node and many small-degree nodes), and also to consider the synchronization dynamics of larger herds.  Additionally, herds of cattle are known to have hierarchies, as not all cows are created equal, and this can be incorporated into the model either through an appropriate network architecture or by considering heterogeneity in the dynamics of individual cows.

An alternative modeling choice would be to consider agent-based models for the herd \cite{franz} rather than the oscillator model that we have studied.  Agent-based formulations are good at incorporating spatial effects, but they of course have a black-box flavor that makes them very difficult to analyze.

The inherent oscillation between the standing/eating phase and the lying/ruminating phase has interesting biological consequences. For example, to stay together as a herd, it is not necessary for all cows to be exactly synchronized, as is sometimes believed. It is possible (and it has been observed often in fields) for a herd to have some individuals lying down and other individuals standing and grazing around them.  From a functional perspective, it is conceivable that this could lead to better spotting of predators than if everyone had their heads down at the same time. A degree of desychronization (provided that it didn't lead to the herd breaking up) might actually be better for each individual than perfect synchronization \cite{beauchamp09}.  Intriguingly the recent model of groups of animals by Dost\'{a}lkov\'{a} and \v{S}pinka, in which each individual can either move or stay in one place, has found evidence (using optimization of a cost function) of partial synchronization but that completely synchronized and completely desynchronized situations seem to occur for a much larger set of parameter values \cite{dostalkova10}.  Moreover, their ``paradoxical" prediction that average group size might decrease as the ratio of the grouping benefit to grouping cost increases is in some sense similar (at least philosophically) to our prediction that less synchronization can potentially occur even with stronger coupling between individual cows.

Although we have framed our discussion in terms of cows, our oscillator framework is very general and should also be useful---perhaps with modifications that are tailored to different species---in studying the behavior of other ruminants.  It is of considerable biological interest to establish empirically which mechanisms for synchrony actually operate in real cows (and, more generally, in other ruminants and in other animals) and to discern more precisely the extent to which such synchrony actually occurs.  It is thus important to develop testable predictions that can help one distinguish the numerous possible synchronization mechanisms.  We have taken one small step in this paper, but there is clearly a lot more interesting research on the horizon.  It is also desirable to consider practical situations, such as the effect of changing pen shape, stocking density, size of lying area, feed-trough size and position, and the nutrient quality of the food.

In addition to the many fascinating animal-behavior questions, the research reported in this paper also suggests several interesting abstract questions.  For example, although the theory of synchronization is well-developed and widely used for smooth dynamical systems~\cite{Pecora_PRL98, Belykh_PHYSICAD06, Arenas_PHYREP08, Sun_SIADS09, Sun_EPL09}, it is an open problem to predict in general when a system that is composed of coupled piecewise smooth oscillators can achieve a stable synchronous state.  In pursuing such considerations, it would also be relevant to consider different notions of synchrony.  Such analysis is of potential importance given the wealth of piecewise smooth dynamical systems that arise in many applications~\cite{Bernardo_2007}. Furthermore, the effects of delay and changes in the network architecture in time are also expected to affect the synchronization properties, though such considerations are difficult even for smooth systems~\cite{Stilwell_SIADS06, Sethia_PRL08, Senthilkumar_PRE09}. We hope that that the model that we have developed in this paper will stimulate research along these lines.



\section{Conclusions} \label{six}

We modeled the eating, lying, and standing dynamics of a cow using a piecewise affine dynamical system.  We constructed Poincar\'e maps to examine the system's equilibrium point and low-period cycles in depth and illustrated more complicated behavior using bifurcation diagrams.  We then considered a model of coupled cows---first using two cows and then using networks of interacting cows---in order to study herd synchrony.  We chose a form of coupling based on cows having an increased desire to eat if they notice another cow eating and an increased desire to lie down if they notice another cow lying down.  We constructed a measure of synchrony that keeps track of when each cow is in a given state and showed that it is possible for cows to synchronize \emph{less} when the coupling is increased.  We also discussed other forms of coupling and cow-interaction networks that can be studied using our formulation.  This line of inquiry seems very promising and that it will not only lead to interesting future theoretical investigations but can even motivate new experiments.  Although we framed our discussion in terms of cows, our framework is general and it should be fruitful in the study of the behavior of other ruminants as well.  The stakes are high when studying animal behavior, and we believe that our model of cattle herds (and generalizations of our model) will yield increased understanding of their synchronization properties.  Milking these ideas as much as possible should prove to be very insightful from both theoretical and practical perspectives.


\section*{Acknowledgments}

We thank Thomas Woolley for useful comments on a draft of this manuscript, Puck Rombach for assistance with cow puns, and Yulian Ng for drawing a spherical cow for us.  JS and EMB gratefully acknowledges partial support from the Army Research Office through grant 51950-MA and the Mathematical Institute of the University of Oxford for hospitality.


\newpage

\appendix


\section{Investigation of the Single Cow Model using Poincar\'e Section} \label{app1}

The single cow model for $w=(x,y;\theta)$, with $(x,y)\in[0,1]\times[0,1]$ and $\theta\in\{\mathcal{E},\mathcal{R},\mathcal{S}\}$, consists of equations describing dynamics for different states $\theta$ and rules for how to switch states.  The equations within each state are
\begin{eqnarray}
	&\mbox{($\cal{E}$)~\textit{Eating state: }}&
	\begin{cases}\label{A_eq:SinglecowE}
		\dot{x} = -\alpha_{2}x\,,\\
		\dot{y} = \beta_{1}y\,,
	\end{cases} \\
	&\mbox{($\cal{R}$)~\textit{Resting state: }}&
	\begin{cases}\label{A_eq:SinglecowR}
		\dot{x} = \alpha_{1}x\,,\\
		\dot{y} = -\beta_{2}y\,,
	\end{cases} \\
	&\mbox{($\cal{S}$)~\textit{Standing state: }}&
	\begin{cases}\label{A_eq:SinglecowS}
		\dot{x} = \alpha_{1}x\,\\
		\dot{y} = \beta_{1}y\,,
	\end{cases} 
\end{eqnarray}
The rules for switching the state $\theta$ are
\begin{equation}\label{A_eq:SinglecowSwitch}
	\theta\rightarrow
	\begin{cases}
		\cal{E}&\mbox{if $\theta\in\{\cal{R,S}\}$ and $x=1$\,,}\\
		\cal{R}&\mbox{if $\theta\in\{\cal{E,S}\}$ and $x<1\,, y=1$\,,}\\
		\cal{S}&\mbox{if $\theta\in\{\cal{E,R}\}$ and $x<1\,, y=\delta$ (or $x=\delta\,, y<1$)\,.}
	\end{cases}
\end{equation}
All of the parameters ($\alpha_{1,2}$ and $\beta_{1,2}$) are positive.  We use the term \textit{single cow equations} to refer collectively to Eqs.~(\ref{A_eq:SinglecowE},\ref{A_eq:SinglecowR},\ref{A_eq:SinglecowS},\ref{A_eq:SinglecowSwitch}).


\subsection{Transversality of the Poincar\'e Section}

As with smooth systems, the term ``flow" in piecewise smooth dynamical systems designates the usual time-parameterized continuous group \cite{Bernardo_2007}.
\begin{definition}[Flow]
The solution to the single cow equations, which we denote by $\phi(t-t_0,w_0)$ for initial condition $w_0$ at time $t_0$, is called a {\rm flow} of the single cow equations.
\end{definition}

The two strips of the boundary of the single cow equations form a set that we denote by $\Sigma$.  It is defined by
\begin{eqnarray}
	\Sigma&\equiv&\{(x,y;\theta)|x=1\,,\delta\leq{y}\leq{1},s=\mathcal{E}\}\cup\{(x,y;\theta)|\delta\leq{x}\leq{1}\,,y=1\,,\theta=\mathcal{R}\}\nonumber\\
	&=&\mathcal{\partial{E}}\cup\mathcal{\partial{R}}\,,
\end{eqnarray}
where we recall that $\partial\mathcal{E}$ and $\partial\mathcal{R}$ are used to represent the two sets $\{(x,y;\theta)|x=1\,,\delta\leq{y}\leq{1},s=\mathcal{E}\}$ and $\{(x,y;\theta)|\delta\leq{x}\leq{1}\,,y=1\,,\theta=\mathcal{R}\}$.

The following lemma shows that the surface $\Sigma$ can be used as a Poincar\'e section for any flow.  This result follows directly from the equations of motion. 

\begin{lemma}[Transversality and Recurrence of $\Sigma$]
For any initial condition $w_0=(x_0,y_0;\theta_0)$ with initial time $t_0$, the flow $\phi(t-t_0,w_0)$ of the single cow equations is transverse to $\Sigma$.  In other words, the direction of the flow (restricted to the $xy$-plane) is not tangent to $\Sigma$ (also restricted to the $xy$-plane). Furthermore, there exists $t>t_0$ such that $\phi(t-t_0,w_0)\in\Sigma$.
\end{lemma}

A similar lemma holds for the extended Poincar\'e section $\Sigma'$, which is defined as
\begin{eqnarray}
	\Sigma'&\equiv&\Sigma\cup\{(x,y;\theta)|x=\delta,\delta\leq{y}<1\}\cup\{(x,y;\theta)|\delta\leq{x}<1\,,y=\delta\}\nonumber\\
	&=&\mathcal{\partial{E}}\cup\mathcal{\partial{R}}\cup\mathcal{\partial{S}}_y\cup\mathcal{\partial{S}}_x\,,
\end{eqnarray}
where $\partial\mathcal{S}_x$ and $\mathcal{S}_y$ are used to represent the sets
$\{(x,y;\theta)|x=\delta,\delta\leq{y}<1\}$ and $\{(x,y;\theta)|\delta\leq{x}<1\,,y=\delta\}$, respectively. 


\subsection{Discrete Dynamics on the Poincar\'e Section: Derivation}

The derivation of map $g$ on $\Sigma'$ involves first solving for the flows on the continuous segments where $\theta$ takes one value. 

Starting from $\theta=\mathcal{E}$, we get
\begin{equation}
	(x,y;\cal{E})\rightarrow
	\begin{cases}
	t_{\mathcal{ER}}=\frac{1}{\beta_1}\log(\frac{1}{y})\,, & 
	g_{\mathcal{ER}}(x,y;\mathcal{E}) = (y^{\frac{\alpha_2}{\beta_1}},1;\mathcal{R})\,,\vspace{0.2cm}\\
	t_{\cal{ES}}=\frac{1}{\alpha_2}\log(\frac{1}{\delta})\,, & 
	g_{\mathcal{ES}}(x,y;\mathcal{E}) = (\delta,(\frac{1}{\delta})^{\frac{\beta_1}{\alpha_2}}y;\mathcal{S})\,.
	\end{cases}
\end{equation}
Starting from $\theta=\mathcal{R}$, we get
\begin{equation}
	(x,y;\mathcal{R})\rightarrow
	\begin{cases}
	t_{\mathcal{RE}}=\frac{1}{\alpha_1}\log(\frac{1}{x})\,, &
	 g_{\mathcal{RE}}(x,y;\mathcal{R}) = (1,{x}^{\frac{\beta_2}{\alpha_1}};\mathcal{E})\,,\vspace{0.2cm}\\
	t_{\mathcal{RS}}=\frac{1}{\beta_2}\log(\frac{1}{\delta})\,, & 
	g_{\mathcal{RS}}(x,y;\mathcal{R}) = ((\frac{1}{\delta})^{\frac{\alpha_1}{\beta_2}}x,\delta;\mathcal{S})\,.
	\end{cases}
\end{equation}
Starting from $\theta=\mathcal{S}$, we get
\begin{equation}
	(x,y;\mathcal{S})\rightarrow
	\begin{cases}
	t_{\mathcal{SE}}=\frac{1}{\alpha_1}\log(\frac{1}{x})\,, &
	 g_{\mathcal{SE}}(x,y;\mathcal{S}) = (1,(\frac{1}{x})^{\frac{\beta_1}{\alpha_1}}y;\mathcal{E})\,,\vspace{0.2cm}\\
	t_{\mathcal{SR}}=\frac{1}{\beta_1}\log(\frac{1}{y})\,, & 
	g_{\mathcal{SR}}(x,y;\mathcal{S}) = ((\frac{1}{y})^{\frac{\alpha_1}{\beta_1}}x,1;\mathcal{R})\,.
	\end{cases}
\end{equation}

Subscripts such as $\mathcal{ER}$ indicate the transition of $\theta$ from one state (e.g., $\mathcal{E}$) to another (e.g., $\mathcal{R}$). The quantity $t$ with the appropriate subscript represents the time it takes for this transition to happen. In the next subsection of this appendix, we will analyze the dependence of the discrete system specified by the above rules on the parameter values and initial conditions.

Using the above equations, we derive the discrete dynamics $g$ on $\Sigma'$ from the $n$th transition to the $(n+1)th$ transition.  This is given by $(x_{n+1},y_{n+1},\theta_{n+1}) = g(x_n,y_n,\theta_n)$, where
\begin{eqnarray}
	g(x=1,\delta\leq{y}\leq{1};\mathcal{E}) = 
	\begin{cases}
	 (y^{\frac{\alpha_2}{\beta_1}},1;\mathcal{R})\,, & 
		\mbox{~if~}y\geq\delta^{\frac{\beta_1}{\alpha_2}}\,,\\
	(\delta,\delta^{-\frac{\beta_1}{\alpha_2}}y;\mathcal{S})\,, &
		\mbox{~if~}y<\delta^{\frac{\beta_1}{\alpha_2}}\,,\\
	\end{cases}\nonumber\\
	g(\delta\leq{x}<1,y=1;\mathcal{R}) = 
	\begin{cases}
	(1,{x}^{\frac{\beta_2}{\alpha_1}};\mathcal{E})\,, &
		\mbox{~if~}x\geq\delta^{\frac{\alpha_1}{\beta_2}}\,,\\
	(\delta^{-\frac{\alpha_1}{\beta_2}}x,\delta;\mathcal{S})\,, &
		\mbox{~if~}x<\delta^{\frac{\alpha_1}{\beta_2}}\,,\\	
	\end{cases}\nonumber\\
	g(x=\delta,\delta\leq{y}<1;\mathcal{S}) = 
	\begin{cases}
		(1,\delta^{-\frac{\beta_1}{\alpha_1}}y;\mathcal{E})\,, &
			\mbox{~if~}y\leq{\delta}^{\frac{\beta_1}{\alpha_1}}\,,\\
		(y^{-\frac{\alpha_1}{\beta_1}}\delta,1;\mathcal{R})\,, &
			\mbox{~if~}y>{\delta}^{\frac{\beta_1}{\alpha_1}}\,,\\
	\end{cases}\nonumber\\
	g(\delta<x<1,y=\delta;\mathcal{S}) = 
	\begin{cases}
		(1,x^{-\frac{\beta_1}{\alpha_1}}\delta;\mathcal{E})\,, &
			\mbox{~if~}x\geq{\delta}^{\frac{\alpha_1}{\beta_1}}\,,\\
		(\delta^{-\frac{\alpha_1}{\beta_1}}x,1;\mathcal{R})\,, &
			\mbox{~if~}x<{\delta}^{\frac{\alpha_1}{\beta_1}}\,.\\
	\end{cases}
\end{eqnarray}
This, in turn, yields the discrete dynamics on $\Sigma'$. The dynamics $f$ on $\Sigma$ is then simply $g$ restricted to $\Sigma$.

We need the following definitions in order to discuss of stability of orbits on the dynamics on $\Sigma$. We start by defining an appropriate distance measure on $\Sigma$.

\begin{definition}[Distance measure on $\Sigma$]
We define the {\rm distance} $\| \cdot \|$ on $\Sigma$ by
\begin{equation}
	\|w_1-w_2\| \equiv |x_1-x_2|+|y_1-y_2|\,,
\end{equation}
where $w_{i}=(x_i,y_i;\theta_i)$ for $i=1,2$.  Note that the symbolic variable $\theta$ does not affect the distance.  
\end{definition}

We now give a definition of stability and asymptotic stability that is analogous to the standard definition for smooth dynamical systems~\cite{Bernardo_2007,Perko_1996}.

\begin{definition}[Stability and Asymptotic Stability]
Let $f$ denote the discrete dynamics on $\Sigma$. A fixed point $w_0$ on $\Sigma$ is {\rm stable} if for any $\epsilon>0$, there exists an $\eta>0$ such that
\begin{equation}
	\|w-w_0\|<\eta\Rightarrow \|f(w)-f(w_0) \|<\epsilon\,.
\end{equation}
A fixed point $w_0$ is {\rm asymptotically stable} if it is stable and there exists an $\eta>0$ such that
\begin{equation}
	\|w-w_0\|<\eta\Rightarrow\lim_{n\rightarrow\infty}f^{n}(w)=w_0\,.
\end{equation}
The stability and asymptotic stability of a period-$T$ orbit $\{z_0,\dots,z_{T-1}\}$ is defined as the stability and asymptotic stability of the fixed point $z_0$ of the $T$-th iterate $f^T$ of the map $f$ on $\Sigma$.
\end{definition}


\subsection{Fixed Points: Existence and Stability}

We first show that there is only one fixed point on $\Sigma$.

\begin{lemma}[Fixed Points on $\Sigma$]
	The only fixed point of $f$ on $\Sigma$ is the point $w_0=(x_0,y_0;\theta)=(1,1;\mathcal{E})$. This fixed point is (locally) asymptotically stable if and only if
	\begin{equation}\label{A_eq:FixedPtCriteria}
		\frac{\alpha_2}{\alpha_1}\cdot\frac{\beta_2}{\beta_1}<1\,.
	\end{equation}
\end{lemma}
\begin{proof}
First we show that if $(x_0,y_0;\theta_0)=(1,y_0;\mathcal{E})$, where $\delta\leq{y_0}<1$ or $(x_0,y_0;\theta_0)=(x_0,1;\mathcal{R})$, then it is not a fixed point. Suppose that there is a fixed point starting from $w_0=(x_0=1,\delta\leq{y_0}<1;\theta_0=\mathcal{E})$. Because it is a fixed point on $\Sigma$, the flow cannot hit $\partial\mathcal{R}$.  It must thus intersect $\partial\mathcal{S}_y$ first and then continue and hit $\partial\mathcal{E}$ again; see Fig.~\ref{appendixfig:FixedPoint} for an illustration.  However, because the $y$-component increases exponentially with rate $\beta_1>0$ when both $\theta=\mathcal{E}$ and $\theta=\mathcal{S}$---see Eqs.~(\ref{A_eq:SinglecowE},\ref{A_eq:SinglecowR},\ref{A_eq:SinglecowS})---it follows that $y_1>y_0$. Consequently, $(x_1,y_1)$ cannot be the same point as $(x_0,y_0)$.  A similar argument applies to initial conditions with $\theta_0=\mathcal{R}$, so we can conclude that there is no fixed point for the discrete dynamics on $\Sigma-\{(1,1,\mathcal{E})\}$.

\begin{figure}[htcp]
\centering
\includegraphics[scale=0.5]{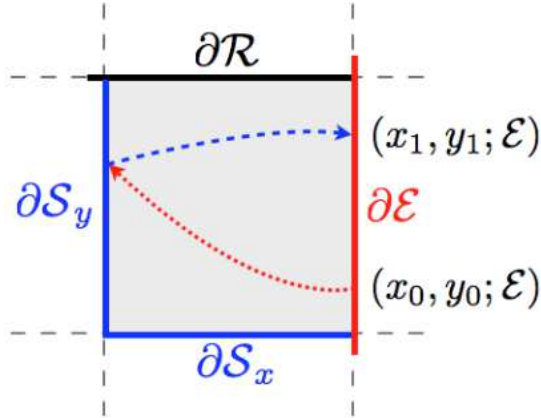}
\caption{(Color online) Illustration that there cannot be any fixed point of $f$ on $\Sigma$ except for the corner point $(1,1;\mathcal{E})$.  This follows from the monotonic increase of the $y$-component when $\theta=\mathcal{E}$ remains unchanged. The other possible situation (not pictured) occurs when $\theta=\mathcal{R}$, for which the $x$-component increases monotonically, indicating that there cannot be an equilibrium point on $\partial\mathcal{R}$.  
}
\label{appendixfig:FixedPoint}
\end{figure}

The only possible fixed point on $\Sigma$ is the point $(1,1;\mathcal{E})$. The asymptotic stability of this fixed point is easily obtained through linearization.
\end{proof}

We remark that although linearization gives local asymptotic stability of the fixed point, numerical simulation indicates that the actual basin of attraction is the entire domain when Eq.~(\ref{A_eq:FixedPtCriteria}) is satisfied.


\subsection{Period-Two Orbits: Existence and Stability}

We next analyze all possible period-two orbits of $f$ on $\Sigma$. Some of those orbits correspond to higher-period orbits of $g$ on $\Sigma'$. When this is the case, we list the points of such an orbit on $\Sigma'$ to differentiate between different periodic orbits.  Nevertheless, it is useful to keep in mind that when restricted to $\Sigma$ (i.e., when one ignores points such that $\theta=\mathcal{S}$), such orbits have period two. Figure~\ref{appendixfig:PeriodTwo} illustrates all of the possible period-two orbits.

\begin{figure}[htcp]
\centering
\includegraphics[scale=0.5]{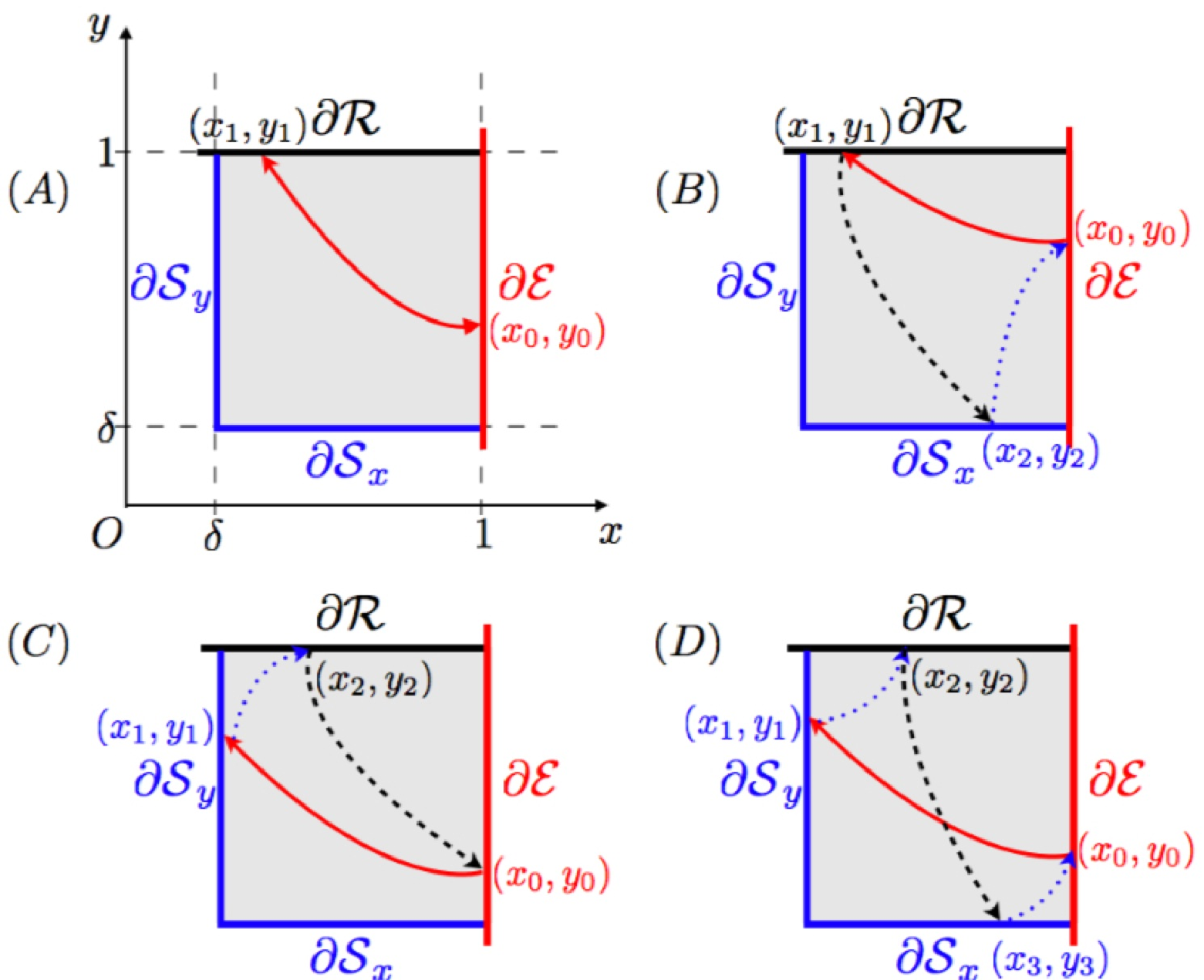}
\caption{(Color online) Illustration of all of the possible period-two orbits on $\Sigma$. 
}
\label{appendixfig:PeriodTwo}
\end{figure}

\subsubsection{Case A:~$(x_0,y_0;\mathcal{E})\rightarrow(x_1,y_1;\mathcal{R})\rightarrow(x_0,y_0;\mathcal{E})\rightarrow\dots$}

This period-two orbit satisfies
\begin{eqnarray}\label{casea}
	w_0 &=& (x_0,y_0;\mathcal{E}) = (1,y_0;\mathcal{E}), \mbox{~where~}0<y_0<1;\nonumber\\
	f(w_0) = g(w_0) = w_1 &=& (x_1,y_1;\mathcal{R}) = (y_0^{\frac{\alpha_2}{\beta_1}},1;\mathcal{R})\,, 
		\mbox{~where~}y_0\geq\delta^{\frac{\beta_1}{\alpha_2}}\,;\nonumber\\
	w_0 = f^2(w_0) = g^2(w_0) = w_2 &=& (x_2,y_2;\mathcal{E}) = (1,x_1^{\frac{\beta_2}{\alpha_1}}\,,\mathcal{E})\,,
		\mbox{~where~}x_1\geq\delta^{\frac{\alpha_1}{\beta_2}}\,.
\end{eqnarray}


The existence of this orbit entails that $(x_0,y_0)=(x_2,y_2)$ and that the constraints (i.e., the inequalities that accompany the equations) are satisfied in (\ref{casea}).  It is thus required that the parameters satisfy
\begin{equation}
	\frac{\alpha_2}{\alpha_1}\cdot\frac{\beta_2}{\beta_1} = 1\,,
\end{equation}
and that $y_0$ satisfy
\begin{equation}
	\begin{cases}
	\delta<y_0<1, & \mbox{~if~}\alpha_2\leq\beta_1\,,\\
	\delta^{\frac{\beta_1}{\alpha_2}}<y_0<1, & \mbox{~if~}\alpha_2>\beta_1\,.
	\end{cases}
\end{equation}
Linearization shows linear stability of the orbit but not necessarily asympotic stability.  

As we are taking $x_0 = 1$ (the initial point is on the right edge of the square domain), we obtain conditions for $y_0$.  All orbits must hit the right edge at some point, so we do not lose any generality by taking $x_0 = 1$.



\subsubsection{Case B:~$(x_0,y_0;\mathcal{E})\rightarrow(x_1,y_1;\mathcal{R})\rightarrow(x_2,y_2;\mathcal{S})\rightarrow(x_0,y_0;\mathcal{E})\rightarrow\dots$}


In this case, the orbit has period two on $\Sigma$ but period three on $\Sigma'$.  Specifically,
\begin{eqnarray}
	w_0 &=& (x_0,y_0;\theta_0) = (1,y_0;\mathcal{E})\,, \mbox{~where~}\delta\leq{y_0}<1;\nonumber\\
	f(w_0) = g(w_0) = w_1 &=& (x_1,y_1;\theta_1) = (y_0^{\frac{\alpha_2}{\beta_1}},1;\mathcal{R})\,, 
		\mbox{~where~}y_0\geq\delta^{\frac{\beta_1}{\alpha_2}}\,;\nonumber\\
	g^2(w_0) = w_2 &=& (x_2,y_2;\theta_2) = (\delta^{-\frac{\alpha_1}{\beta_2}}x_1,\delta;\mathcal{S})\,,
		\mbox{~where~}x_1<\delta^{\frac{\alpha_1}{\beta_2}}\,;\nonumber\\
	w_0 = f^2(w_0) = g^3(w_0) = w_3 &=& (x_3,y_3;\theta_3) = (1,x_2^{-\frac{\beta_1}{\alpha_1}}\delta;\mathcal{E})\,,
		\mbox{~where~}x_2\geq\delta^{\frac{\alpha_1}{\beta_1}}\,.
\end{eqnarray}
For fixed parameter values, there is only one such orbit; it must satisfy
\begin{equation}
	y_0 = \delta^{\frac{1+\frac{\beta_1}{\beta_2}}{1+\frac{\alpha_2}{\alpha_1}}}\,.
\end{equation}

The existence of this period-two orbit also requires that the parameters satisfy
\begin{align}
	\frac{\alpha_2}{\alpha_1}\cdot\frac{\beta_2}{\beta_1}&>1 \notag \\
		\frac{1}{\alpha_1}+\frac{1}{\alpha_2}&\geq\frac{1}{\beta_1}+\frac{1}{\beta_2}
		\mbox{~if~}\beta_1<\alpha_2\,.
\end{align}
This orbit is asymptotically stable if and only if
\begin{equation}
	\frac{\alpha_2}{\alpha_1}<1\,.
\end{equation}
In particular, it is worth remarking that the orbit is not asymptotically stable in the case $\alpha_1=\alpha_2$ describing equal growth and decay rates for hunger.


\subsubsection{Case C:~$(x_0,y_0;\mathcal{E})\rightarrow(x_1,y_1;\mathcal{S})\rightarrow(x_2,y_2;\mathcal{R})\rightarrow(x_0,y_0;\mathcal{E})\rightarrow\dots$}

In this case, the orbit has period two on $\Sigma$ but period four on $\Sigma'$.  Specifically,
\begin{eqnarray}\label{casec}
	w_0 &=& (x_0,y_0;s_0) = (1,y_0;\mathcal{E})\,,
		\mbox{~where~}\delta\leq{y_0}<1;\nonumber\\
	g(w_0) = w_1 &=& (x_1,y_1;s_1) = (\delta,\delta^{-\frac{\beta_1}{\alpha_2}}y_0;\mathcal{S})\,,
		\mbox{~where~}y_0<\delta^{\frac{\beta_1}{\alpha_2}}\,;\nonumber\\
	f(w_0) = g^2(w_0) = w_2 &=& (x_2,y_2;s_2) = (y_1^{-\frac{\alpha_1}{\beta_1}}\delta,1;\mathcal{R})\,,
		\mbox{~where~}y_1>\delta^{\frac{\beta_1}{\alpha_1}}\,;\nonumber\\
	w_0 = f^2(w_0) = g^3(w_0) = w_3 &=& (x_3,y_3;s_3) = (1,x_2^{\frac{\beta_2}{\alpha_1}};\mathcal{E})\,,
		\mbox{~where~}x_2\geq\delta^{\frac{\alpha_1}{\beta_2}}\,.
\end{eqnarray}

Solving (\ref{casec}) with the associated constraints yields necessary conditions for the existence of this period-two orbit.  The initial value $y_0$ must satisfy
\begin{equation}
	y_0 = \delta^{\frac{\frac{1}{\alpha_1}+\frac{1}{\alpha_2}}{\frac{1}{\beta_1}+\frac{1}{\beta_2}}}\,,
\end{equation}
and the parameters must satisfy
\begin{equation}
	\begin{cases}
	\frac{\alpha_2}{\alpha_1}\cdot\frac{\beta_2}{\beta_1}>1\,,\\
	\beta_1<\alpha_2\,,\\
	\frac{1}{\alpha_1}+\frac{1}{\alpha_2}<\frac{1}{\beta_1}+\frac{1}{\beta_2}\,,\\
	\frac{\alpha_2}{\alpha_1}\leq\frac{\beta_2}{\beta_1}\,.
	\end{cases}
\end{equation}
This orbit is asymptotically stable if and only if
\begin{equation}
	\frac{\beta_2}{\beta_1}<1\,.
\end{equation}
Note, in particular, that this implies that the orbit is not asymptotically stable when $\beta_1 = \beta_2$ (i.e., when the growth and decay rates for desire to lie down are equal).


\subsubsection{Case D:~$(x_0,y_0;\mathcal{E})\rightarrow(x_1,y_1;\mathcal{S})\rightarrow(x_2,y_2;\mathcal{R})\rightarrow(x_3,y_3;\mathcal{S})\rightarrow(x_0,y_0;\mathcal{E})\rightarrow\dots$}

In this case, the orbit has period two on $\Sigma$ but period four on $\Sigma'$.  Specifically,
\begin{eqnarray}
	w_0 &=& (x_0,y_0;s_0) = (1,y_0;\mathcal{E})\,,
		\mbox{~where~}\delta\leq{y_0}<1;\nonumber\\
	g(w_0) = w_1 &=& (x_1,y_1;s_1) = (\delta,\delta^{-\frac{\beta_1}{\alpha_2}}y_0;\mathcal{S})\,,
		\mbox{~where~}y_0<\delta^{\frac{\beta_1}{\alpha_2}}\,;\nonumber\\
	f(w_0) = g^2(w_0) = w_2 &=& (x_2,y_2;s_2) = (y_1^{-\frac{\alpha_1}{\beta_1}}\delta,1;\mathcal{R})\,,
		\mbox{~where~}y_1>\delta^{\frac{\beta_1}{\alpha_1}}\,;\nonumber\\
	g^3(w_0) = w_3 &=& (x_3,y_3;s_3) = (\delta^{-\frac{\alpha_1}{\beta_2}}x_2,\delta;\mathcal{S})\,,
		\mbox{~where~}x_2<\delta^{\frac{\alpha_1}{\beta_2}}\,;\nonumber\\
	w_0 = f^2(w_0) = g^4(w_0) = w_4 &=& (x_4,y_4;s_4) = (1,x_3^{-\frac{\beta_1}{\alpha_1}}\delta;\mathcal{E})\,,
		\mbox{~where~}x_3\geq\delta^{\frac{\alpha_1}{\beta_1}}\,.\nonumber\\
\end{eqnarray}

The existence of such orbits entails that
\begin{eqnarray}
	\frac{\alpha_2}{\alpha_1}\cdot\frac{\beta_2}{\beta_1}&>&1\,,\nonumber\\
	\frac{1}{\alpha_1}+\frac{1}{\alpha_2}&=&\frac{1}{\beta_1}+\frac{1}{\beta_2}\,, \notag \\
	\beta_1&<&\alpha_2\,.
\end{eqnarray}
This yields infinitely many such orbits, for which $x_0=1$ and 
\begin{equation}
	\delta<y_0<\delta^{\frac{\beta_1}{\alpha_2}}\,.
\end{equation}
All of these orbits are stable but not asymptotically stable.




\begin{thebibliography}{10}

\bibitem{arave76}
{\sc C.~W. Arave and J.~L. Albright}, {\em Social rank and physiological traits
  of dairy cows as influenced by changing group membership.}, Journal of Dairy
  Science, 59 (1976), pp.~974--981.

\bibitem{Arenas_PHYREP08}
{\sc A.~Arenas, A.~D\'{i}az-Guilera, J.~Kurths, Y.~Moreno, and C.~Zhou}, {\em
  Synchronization in complex networks}, Physics Reports, 469 (2008),
  pp.~93--153.

\bibitem{beauchamp09}
{\sc G.~Beauchamp}, {\em How does food density influence vigilance in birds and
  mammals?}, Animal Behaviour, 78 (2009), pp.~223--231.

\bibitem{Belykh_PHYSICAD06}
{\sc I.~Belykh, V.~Belykh, and M.~Hasler}, {\em Generalized connection graph
  method for synchronization in asymmetrical networks}, Physica D, 224 (2006),
  pp.~42--51.

\bibitem{ebj94}
{\sc E.~Ben-Jacob, O.~Shochet, A.~Tenenbaum, I.~Cohen, A.~Czir\'{o}k, and
  T.~Vicsek}, {\em Generic modelling of cooperative growth patterns in
  bacterial colonies}, Nature, 368 (1994), pp.~46--49.

\bibitem{benham82}
{\sc P.~F.~J. Benham}, {\em Synchronisation of behaviour in grazing cattle},
  Applied Animal Ethology, 8 (1982), pp.~403--404.

\bibitem{boe06}
{\sc K.~E. B{\o}e, S.~Berg, and I.~L. Andersen}, {\em Resting behaviour and
  displacement in ewes---effects of reduced lying space and pen shape}, Applied
  Animal behaviour Science, 98 (2006), pp.~249--259.

\bibitem{scholarpiecewise}
{\sc A.~R. Champneys and M.~{di Bernardo}}, {\em Piecewise smooth dynamical
  systems}, Scholarpedia, 3 (2008).
\newblock {\tt
  http://www.scholarpedia.org/article/Piecewise\_smooth\_dynamical\_systems}.

\bibitem{chuang07}
{\sc Y.-L. Chuang, Y.~R. Huang, M.~R. D'Orsogna, and A.~L. Bertozzi}, {\em
  Multi-vehicle flocking: {S}calability of cooperative control algorithms using
  pairwise potentials}, IEEE International Conference on Robotics and
  Automation,  (2007), pp.~2292--2299.

\bibitem{conradt09}
{\sc L.~Conradt and C.~List}, {\em Group decisions in humans and animals},
  Philosophical Transactions of the Royal Society B, 364 (2009), pp.~719--742.

\bibitem{conradt00}
{\sc L.~Conradt and T.~J. Roper}, {\em Activity synchrony and social cohesion:
  a fission-fusion model}, Proceedings of the Royal Society B, 267 (2000),
  pp.~2213--2218.

\bibitem{couzinreview}
{\sc I.~D. Couzin}, {\em Collective cognition in animal groups}, Trends in
  Cognitive Science, 13 (2009), pp.~36--43.

\bibitem{couzin}
{\sc I.~D. Couzin, J.~Krause, R.~James, G.~D. Ruxton, and N.~R. Franks}, {\em
  Collective memory and spatial sorting in animal groups}, Journal of
  Theoretical Biology, 218 (2002), pp.~1 -- 11.

\bibitem{couzininpress}
{\sc I.~D. Couzin, T.~Murphy, and V.~Guttal}, {\em Living architecture:
  Self-organized bridge construction in army ants}.
\newblock in preparation.

\bibitem{predrag}
{\sc P.~Cvitanovi\'{c}, R.~Artuso, R.~Mainieri, G.~Tanner, and G.~Vattay}, {\em
  Chaos: Classical and Quantum}, Niels Bohr Institute, Copenhagen, 13~ed.,
  2009.
\newblock {C}haos{B}ook.org.

\bibitem{Bernardo_2007}
{\sc M.~{di Bernardo}, C.~J. Budd, A.~R. Champneys, and P.~Kowalczyk}, {\em
  Piecewise-smooth Dynamical Systems: Theory and Applications},
  Springer-Verlag, Berlin, Germany, 1~ed., December 2007.

\bibitem{dostalkova10}
{\sc I.~Dost\'{a}lkov\'{a} and M.~\v{S}pinka}, {\em When to go with the crowd:
  modelling synchronisation of all-or-nothing activity transitions in groups
  animals}, Journal of Theoretical Biology, 263 (2010), pp.~437--448.

\bibitem{dyer09}
{\sc J.~R.~G. Dyer, A.~Johansson, D.~Helbing, I.~D. Couzin, and J.~Krause},
  {\em Leadership, consensus decision making and collective behaviour in
  humans}, Philosophical Transactions of the Royal Society B: Biological
  Sciences, 364 (2009), pp.~781--789.

\bibitem{estevez07}
{\sc I.~Estevez, I.~L. Andersen, and E.~N\ae~vdal}, {\em Group size, density
  and social dynamics in farm animals}, Applied Animal Behaviour, 103 (2007),
  pp.~185--204.

\bibitem{faerevik08}
{\sc G.~F{\ae}revik, K.~Tjentland, S.~L{\o}vik, I.~L. Andersen, and K.~E.
  B{\o}e}, {\em Resting pattern and social behaviour of dairy calves housed in
  pens with different sized lying areas}, Applied Animal Behaviour Science, 114
  (2008), pp.~54--64.

\bibitem{fisher02}
{\sc A.~D. Fisher, G.~A. Verkerk, C.~J. Morrow, and L.~R. Matthews}, {\em The
  effects of feed restriction and lying deprivation on pituitary adrenal axis
  regulation in lactating cows}, Livestock Production Science, 73 (2002),
  pp.~255--263.

\bibitem{franz}
{\sc B.~Franz}, {\em Synchronisation properties of an agent-based animal
  behaviour model}.
\newblock MSc dissertation, University of Oxford, September 2009.

\bibitem{Glass_1975}
{\sc L.~Glass}, {\em Combinatorial and topological methods in nonlinear
  chemical kinetics}, The Journal of Chemical Physics, 63 (1975),
  pp.~1325--1335.

\bibitem{Gouze_2002}
{\sc {J.-L.} Gouz\'e and T.~Sari}, {\em A class of piecewise linear
  differential equations arising in biological models}, Dynamical Systems, 17
  (2002), pp.~299--316.

\bibitem{gygax07}
{\sc L.~Gygax, C.~Mayer, H.~S. Westerath, K.~Friedli, and B.~Wechslerf}, {\em
  On-farm assessment of the lying behaviour of finishing bulls kept in housing
  systems with different floor qualities.}, Animal Welfare, 16 (2007),
  pp.~205--208.

\bibitem{hindhede96}
{\sc J.~Hindhede, J.~T. S{\o}rensen, M.~B. Jensen, and C.~C. Krohn}, {\em
  Effects of space allowance, access to bedding and flock size on slatted floor
  system on production and health of dairy heifers}, Acta Agriculturae
  Scandinavica, Section A -- Animal Science, 46 (1996), pp.~46--53.

\bibitem{jensen95}
{\sc M.~B. Jensen and R.~Kyhn}, {\em Play behaviour in group housed dairy
  calves kept in pens; the effect of social contact and space allowance},
  Applied Animal Behaviour Science, 56 (1995), pp.~97--108.

\bibitem{Kowalczyk_2005}
{\sc P.~Kowalczyk}, {\em Robust chaos and border-collision bifurcations in
  non-invertible piecewise-linear maps}, Nonlinearity, 18 (2005), pp.~485--504.

\bibitem{mendl01}
{\sc M.~Mendl and S.~Held}, {\em Living in groups: an evolutionary
  perspective}, in Social Behaviour in Farm Animals, L.~J. Keeling and H.~W.
  Gonyou, eds., CABI, Wallingford, UK, 2001, pp.~7--36.

\bibitem{mogensen97}
{\sc L.~Mogensen, C.~C. Krohn, J.T. S{\o}rensen, J.~Hindhede, and L.~H.
  Nielsen}, {\em Association between resting behaviour and live weight gain in
  dairy heifers housed in pens with different space allowance and floor type},
  Applied Animal behaviour Science, 55 (1997), pp.~11--19.

\bibitem{munk05}
{\sc L.~Munksgaard, M.~B. Jensen, L.~J. Pedersen, S.~W. Hansen, and
  L.~Matthews}, {\em Quantifying behavioural priorities---effects of time
  constraints on behaviour of dairy cows, {B}os {T}aurus}, Applied Animal
  Behaviour Science, 92 (2005), pp.~3--14.

\bibitem{nap09}
{\sc F.~Napolitano, U.~Knierem, F.~Grasso, and G.~{De Rosa}}, {\em Positive
  indicators of cattle welfare and their applicability to on-farm proitocols},
  Italian Journal of Animal Science (Special Issue Supplement 1), 8 (2009),
  pp.~355--365.

\bibitem{nielsen97}
{\sc L.~H. Nielsen, L.~Mogensen, C.~Krohn, Hindhede J., and Sorensen~J. T.},
  {\em Resting and social behaviour of dairy heifers housed in slatted floor
  pens with different sized bedded lying areas}, Applied Animal Behaviour
  Science, 54 (1997), pp.~307--316.

\bibitem{nish10}
{\sc T.~Nishikawa and A.~E. Motter}, {\em Resolving the network synchronization
  landscape: {C}ompensatory structures, quantization, and the positive effect
  of negative interactions}, Proceedings of the National Academy of Sciences,
  (2010).
\newblock to appear (arXiv:0909.2874).

\bibitem{odrisc09}
{\sc K.~{O'Driscoll}, I.~Boyle, and A.~Hanlon}, {\em The effect of breed and
  housing system on dairy cow feeding and lying behaviour}, Applied Animal
  Behaviour Science, 116 (2009), pp.~156--162.

\bibitem{paley07}
{\sc D.~A. Paley, N.~E. Leonard, R.~Sepulchre, D.~Grunbaum, and J.~K. Parrish},
  {\em Oscillator models and collective motion: {S}patial patterns in the
  dynamics of engineered and biological networks}, IEEE Control Systems
  Magazine, 27 (2007), pp.~89--105.

\bibitem{Pecora_PRL98}
{\sc L.~M. Pecora and T.~L. Carroll}, {\em Master stability functions for
  synchronized coupled systems master stability functions for synchronized
  coupled systems master stability functions for synchronized coupled systems},
  Physical Review Letters, 80 (1998), pp.~2109--2112.

\bibitem{Perko_1996}
{\sc L.~Perko}, {\em Differential Equations and Dynamical Systems},
  Springer-Verlag, Berlin, Germany, 2~ed., 1996.

\bibitem{pikovsky}
{\sc A.~Pikovsky, M.~Rosenblum, and J.~Kurths}, {\em Synchronization: A
  Universal Concept in Nonlinear Sciences}, Cambridge University Press, UK,
  2003.

\bibitem{rook91}
{\sc A.~J. Rook and P.~D. Penning}, {\em Synchronisation of eating, ruminating
  and idling activity by grazing sheep}, Applied Animal Behaviour Science, 32
  (1991), pp.~157--166.

\bibitem{scardovi08}
{\sc L.~Scardovi, N.~E. Leonard, and R.~Sepulchre}, {\em Stabilization of three
  dimensional collective motion}, Communications in Information and Systems, 8
  (2008), pp.~473--500.

\bibitem{Senthilkumar_PRE09}
{\sc D.~V. Senthilkumar, J.~Kurths, and M.~Lakshmanan}, {\em Stability of
  synchronization in coupled time-delay systems using krasovskii-lyapunov
  theory}, Physical Review E, 79 (2009), 066208.

\bibitem{Sethia_PRL08}
{\sc G.~C. Sethia, A.~Sen, and F.~M. Atay}, {\em Clustered chimera states in
  delay-coupled oscillator systems}, Physical Review Letters, 100 (2008),
  144102.

\bibitem{tom}
{\sc T.~Shaw}, {\em Personal communication}.
\newblock 2009.

\bibitem{sinai}
{\sc Y.~Sinai}, {\em {WHAT IS}...a billiard}, Notices of the American
  Mathematical Society, 51 (2004), pp.~412--413.

\bibitem{Stilwell_SIADS06}
{\sc D.~J. Stilwell, E.~M. Bollt, and D.~G. Roberson}, {\em Sufficient
  conditions for fast switching synchronization in time-varying network
  topologies}, SIAM Journal on Applied Dynamical Systems, 5 (2006),
  pp.~140--156.

\bibitem{sync}
{\sc S.~H. Strogatz}, {\em {SYNC}: The Emerging Science of Spontaneous Order},
  Hyperion, New York, NY, USA, 2003.

\bibitem{sumpter08}
{\sc D.~J.~T. Sumpter, J.~Krause, R.~James, I.~D. Couzin, and A.~J.~W. Ward},
  {\em Consensus decision making by fish}, Current Biology, 18 (2008),
  pp.~1773--1777.

\bibitem{Sun_SIADS09}
{\sc J.~Sun, E.~M. Bollt, and T.~Nishikawa}, {\em Constructing generalized
  synchronization manifolds by manifold equation}, SIAM Journal on Applied
  Dynamical Systems, 8 (2009), pp.~202--221.

\bibitem{Sun_EPL09}
\leavevmode\vrule height 2pt depth -1.6pt width 23pt, {\em Master stability
  functions for coupled nearly identical dynamical systems}, Europhysics
  Letters, 85 (2009), 60011.

\bibitem{vicsek99}
{\sc T.~Vicsek, A.~Czir\'{o}k, I.~J. Farkas, and D.~Helbing}, {\em Application
  of statistical mechanics to collective motion in biology}, Physica A, 274
  (1999), pp.~182--189.

\bibitem{yates09}
{\sc C.~A. Yates, R.~Erban, C.~Escudero, I.~D. Couzin, J.~Buhl, I.~G.
  Kevrekidis, P.~K. Maini, and D.~J.~T. Sumpter}, {\em Inherent noise can
  facilitate coherence in collective swarm motion}, Proceedings of the National
  Academy of Sciences, 106 (2009), pp.~5464--5469.

\end{thebibliography}


\end{document}